\documentclass[11pt, letterpaper]{article}

\usepackage{booktabs}   
\usepackage{subcaption} 
\usepackage{libertine}
\usepackage{balance}
\usepackage{url}

\usepackage[utf8]{inputenc}
\usepackage{enumitem}
\usepackage{listings}
\usepackage{amsmath}
\usepackage{amsthm}
\usepackage{dsfont}

\usepackage{xspace}
\usepackage{amssymb}
\usepackage{multirow}
\usepackage{tabularx}
\usepackage{longtable}
\usepackage[ruled,linesnumbered,noresetcount]{algorithm2e}
\usepackage{multicol}
\usepackage{etoolbox}
\AtBeginEnvironment{algorithm}{\SetArgSty{textrm}} 

\SetCommentSty{mycommfont}
\newcommand{\myvar}[1]{$\mathit{#1}$}
\newcommand{\var}[1]{\mathit{#1}}
\SetKwComment{Comment}{$\triangleright$\ }{}
\SetEndCharOfAlgoLine{}
\SetKwProg{Procedure}{Procedure}{:}{end}
\SetKwRepeat{Repeat}{do}{while}

\newcommand{\mysketch}{Quancurrent\xspace}

\usepackage{lipsum}
\usepackage{graphicx}
\usepackage{caption}
\usepackage{float}
\usepackage{placeins} 
\usepackage{flafter}
\usepackage{wrapfig}
\graphicspath{ {./images/sequential_quantile/}{./images/Quancurrent/}{./images/graphs/}{./images/} }

\newtheorem{theorem}{Theorem}
\newtheorem{lemma}[theorem]{Lemma}
\newtheorem{claim}[theorem]{Claim}

\newtheorem{invariant}{Invariant}
\newtheorem{definition}{Definition}

\newcommand{\s}{$\sigma$\xspace}
\newcommand{\fs}{\(f(\sigma)\)\xspace}
\newcommand{\fstag}{\(f(\sigma')\)\xspace}
\newcommand{\ls}{\(l(\sigma)\)\xspace}

\newcommand{\Q}{\(\mathnormal{Q}\)\xspace}

\newcommand{\A}{$A=\mathcal{S}(f(\sigma))$\xspace}

\renewcommand{\v}{\emph{v}\xspace}
\newcommand{\probP}{\text{I\kern-0.15em P}}

\usepackage{xcolor}

\newcounter{mycounter}
\usepackage{titlesec}
\usepackage{siunitx}
\ExplSyntaxOn
\int_new:N \l__nebu_min_prefix_int
\tl_new:N \l__nebu_base_number_tl
\tl_new:N \l__nebu_mode_tl
\cs_new:Npn \nebu_prefix:n #1
  {
    \exp_args:Nf \__nebu_prefix:n
      { \exp_args:Nf \fp_to_scientific:n { \tl_lower_case:n {#1} } }
  }
\cs_new:Npn \__nebu_prefix:n #1
  { \__nebu_prefix:nwnw #1 \q_stop }
\cs_new:Npn \__nebu_prefix:nwnw #1 e #2 \q_stop
  { \__nebu_find_prefix:nnn {#2} {#1} {1} }
\cs_new:Npn \__nebu_find_prefix:nnn #1 #2 #3
  {
    \tl_if_exist:cT { l__nebu_ #1 \tl_use:N \l__nebu_mode_tl _prefix_tl }
      {
        \use_i_delimit_by_q_stop:nw
          { \__nebu_output:nnn {#2} {#1} {#3} }
      }
    \int_compare:nNnT {#1} < \l__nebu_min_prefix_int
      {
        \use_i_delimit_by_q_stop:nw
          {
            \exp_args:Nf \__nebu_find_prefix:nnn
              { \int_eval:n { \l__nebu_min_prefix_int } } {#2}
              { #3 / 1\prg_replicate:nn { \l__nebu_min_prefix_int - #1 }{ 0 } }
          }
      }
    \use_i:nn
      {
        \exp_args:Nf \__nebu_find_prefix:nnn
          { \int_eval:n {#1-1} } {#2} {#3*10}
      }
      \q_stop
  }
\exp_args_generate:n { fv }
\cs_new:Npn \__nebu_output:nnn #1 #2 #3
  {
    \exp_args:Nfv \nebu_output:nn
      { \fp_to_decimal:n {#1*#3} } { l__nebu_ #2 \tl_use:N \l__nebu_mode_tl _prefix_tl }
  }
\cs_new_protected:Npn \nebu_prefix_set:Nnn #1 #2 #3
  {
    \exp_args:Nxx \__nebu_prefix_set:nnN
      { \tl_trim_spaces:n {#2} } { \int_eval:n {#3} } #1
  }
\cs_new_protected:Npn \__nebu_prefix_set:nnN #1 #2 #3
  {
    \tl_clear_new:c { l__nebu_ #2 _prefix_tl }
    \tl_clear_new:c { l__nebu_ #2 _siunitx_prefix_tl }
    \tl_set:cn { l__nebu_ #2 _prefix_tl } {#1}
    \tl_set:cn { l__nebu_ #2 _siunitx_prefix_tl } {#3}
    \int_set:Nn \l__nebu_min_prefix_int { \int_min:nn {#2} { \l__nebu_min_prefix_int } }
  }
\NewExpandableDocumentCommand \prefix { m }
  { \nebu_prefix:n {#1} }
\cs_new:Npn \nebu_output:nn #1 #2 { #1\, \textrm{#2} }
\NewDocumentCommand \setprefix { m m m }
  { \nebu_prefix_set:Nnn #1 {#2} {#3} }
\DeclareSIPrefix \none { } { 0 }
\setprefix \none { } { 0 }
\cs_new_protected:Npn \__nebu_store:nn #1 #2
  {
    \tl_set:Nn \l__nebu_base_number_tl {#1}
    \cs_set:Npn \prefix {#2}
  }
\NewExpandableDocumentCommand \prefixSI { o m o m }
  {
    \group_begin:
      \cs_set_eq:NN \nebu_output:nn \__nebu_store:nn
      \tl_set:Nn \l__nebu_mode_tl { _siunitx }
      \nebu_prefix:n {#2}
      \SI [#1] { \l__nebu_base_number_tl } [#3] {#4}
    \group_end:
  }
\ExplSyntaxOff

\setprefix \yocto { y } { -24 }
\setprefix \zepto { z } { -21 }
\setprefix \atto  { a } { -18 }
\setprefix \femto { f } { -15 }
\setprefix \pico  { p } { -12 }
\setprefix \nano  { n } { -9 }
\setprefix \micro { \SIUnitSymbolMicro } { -6 }
\setprefix \milli { m } { -3 }
\setprefix \centi { c } { -2 }
\setprefix \deci  { d } { -1 }
\setprefix \deca  { da } { 1 }
\setprefix \hecto { h }  { 2 }
\setprefix \kilo  { k }  { 3 }
\setprefix \mega  { M }  { 6 }
\setprefix \giga  { G }  { 9 }
\setprefix \tera  { T }  { 12 }
\setprefix \peta  { P }  { 15 }
\setprefix \exa   { E }  { 18 }
\setprefix \zetta { Z }  { 21 }
\setprefix \yotta { Y }  { 24 }

\begin{document}

\title{Quancurrent: A Concurrent Quantiles Sketch}

\author{Shaked Elias-Zada \\ Technion \and Arik Rinberg \\ Technion \and Idit Keidar \\ Technion}

\date{}

\maketitle

    Sketches are a family of streaming algorithms widely used in the world of big data to perform fast, real-time analytics. A popular sketch type is Quantiles, which estimates the data distribution of a large input stream. We present Quancurrent, a highly scalable concurrent Quantiles sketch. Quancurrent’s throughput increases linearly with the number of available threads, and with $32$ threads, it reaches an update speedup of $12$x and a query speedup of $30$x over a sequential sketch. Quancurrent allows queries to occur concurrently with updates and achieves an order of magnitude better query freshness than existing scalable solutions.

\section {Introduction}
Data sketches, or \emph{sketches} for short, are indispensable tools for performing analytics on high-rate, high-volume data. Specifically, understanding the data distribution is a fundamental task in data management and analysis, used in applications such as exploratory data analysis~\cite{vartak2015seedb}, operations monitoring~\cite{abraham2013scuba}, and more. 

The Quantiles sketch family captures this task~\cite{masson2019ddsketch, mergeables_summaries, gan2018moment, cormode2021relative}. The sketch represents the quantiles distribution in a stream of elements, such that for any $0 \leq \phi \leq 1$, a query for quantile $\phi$ returns an estimate of the $\lfloor n\phi \rfloor ^{\text{th}}$ largest element in a stream of size $n$. For example, quantile $\phi=0.5$ is the median. Due to the importance of quantiles approximation, Quantiles sketches are a part of many analytics platforms, e.g., Druid~\cite{druid-quantiles}, Hillview~\cite{budiu2019hillview}, Presto~\cite{presto}, and Spark~\cite{spark}. 

Sketches are designed for \emph{stream} settings, in which each element is processed once. Like other sketches, Quantiles sketch is of sublinear-size and its estimates are \emph{probably approximately correct (PAC)}, providing an approximation within some error $\epsilon n$ with a failure probability bounded by some parameter $\delta$. 

The classic literature on sketches has focused on inducing a small error while using a small memory footprint, in the context of sequential processing: the sketch is built by a single thread, and queries are served only after sketch construction is complete. Only recently, we begin to see works leveraging parallel architectures to achieve a higher ingestion throughput while also enabling queries concurrently with updates~\cite{Rinberg_2020_fast_sketches, stylianopoulos2020delegation}. 
Of these, the only solution suitable for quantiles that we are aware of is the Fast Concurrent Data Sketches (FCDS) framework proposed by Rinberg et al.~\cite{Rinberg_2020_fast_sketches}. However, the scalability of FCDS-based quantiles sketches is limited unless query freshness is heavily compromised (as we show below). Our goal is to provide a scalable concurrent Quantiles sketch that retains a small error bound with reasonable query freshness.

In Section~\ref{sec:background}, we define the problem and overview a popular sequential solution proposed by Agarwal et. al~\cite{mergeables_summaries} which is used by Apache DataSketches~\cite{DataSketches}, on which our concurrent sketch is based.
In Section~\ref{sec:Quancurrent}, we present \mysketch, our highly scalable concurrent Quantiles sketch.
Like FCDS, \mysketch relies on local buffering of stream elements, which are then propagated in bulk to a shared sketch.
But \mysketch improves on FCDS by eliminating the latter's sequential propagation bottleneck, which mostly stems from the need to sort large buffers.

In \mysketch, sorting occurs at three levels – a small thread-local buffer, an intermediate NUMA-node-local buffer called $\mathit{Gather\&Sort}$, and the shared sketch.
Moreover, the shared sketch itself is organized in multiple levels, which may be propagated (and sorted) concurrently by multiple threads.

To allow queries to scale as well, \mysketch serves them from a cached snapshot of the shared sketch.
This architecture is illustrated in Figure~\ref{fig:quancurrentDS}.
The query freshness depends on the sizes of local and NUMA-local buffers as well as the frequency of caching queries.
We show that using this architecture, high throughput can be achieved with much smaller buffers (hence much better freshness) than in FCDS.


\begin{figure}[]
    \centering
    \includegraphics[width=0.8\linewidth,trim={0cm 0cm 0cm 0.3cm},clip]{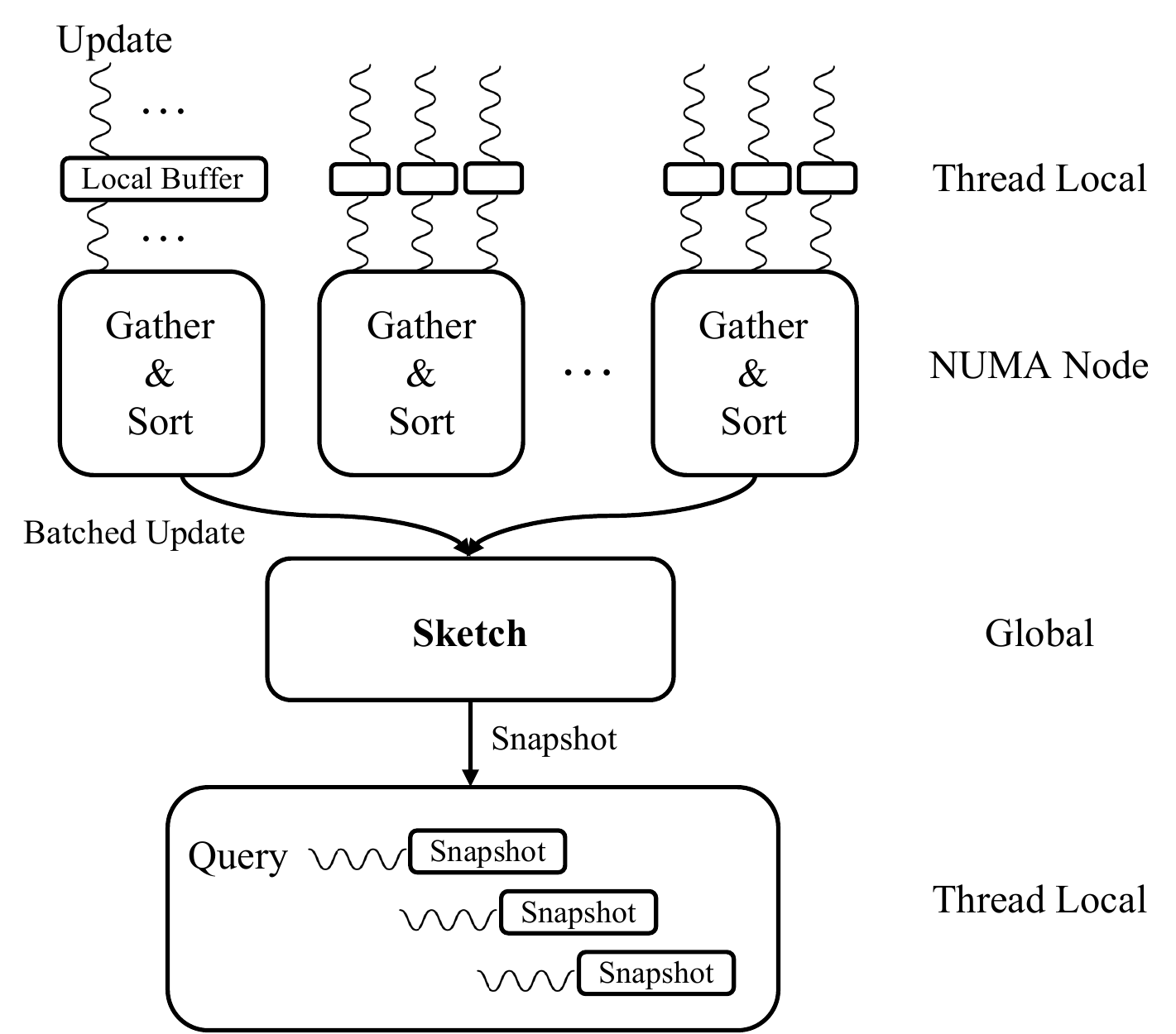}
    \caption{\mysketch's data structures.}
    \label{fig:quancurrentDS}
\end{figure}


To lower synchronization overhead, we allow buffered elements to be sporadically overwritten by others without being propagated, and others to be duplicated, i.e., propagated more than once. These occurrences, which we call \emph{holes}, alter the stream ingested by the data structure. 
Yet, in Section~\ref{sec:analysis} we showed that for a sufficiently large local buffer, the expected number of holes is less than $1$ and, because they are random, they do not change the sampled distribution.
Figure~\ref{fig:intro-query-accuracy} presents quantiles estimated by \mysketch on a stream of normally distributed random values compared to an exact, brute-force computation of the quantiles, and shows that the estimation remains accurate.

\begin{figure}[h]
    \centering
    \includegraphics[width=\linewidth,trim={0cm 0.3cm 0cm 1.5cm},clip]
    {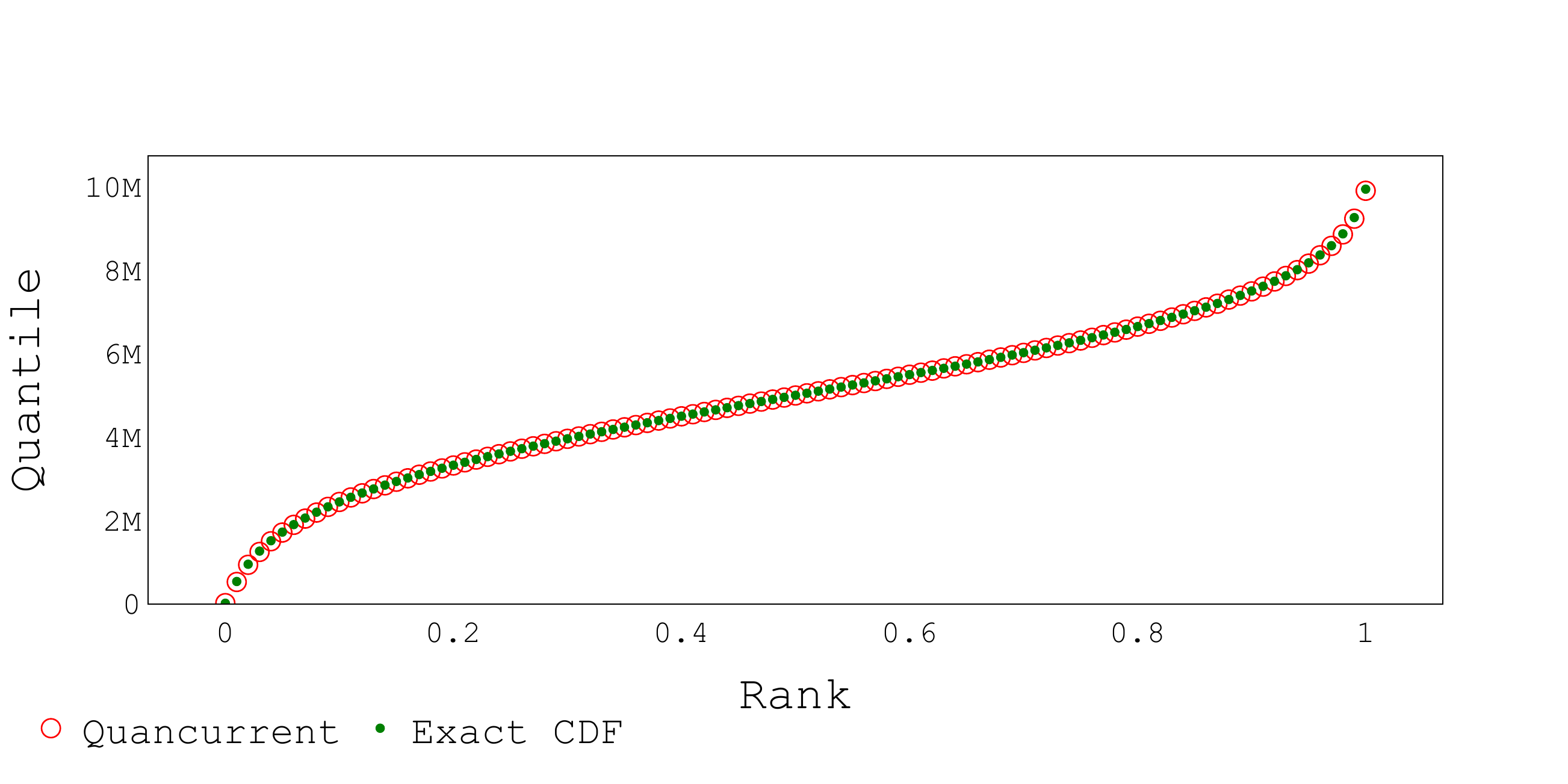}
    \caption{\mysketch quantiles vs. exact CDF, k = 1024, normal distribution, 32 update threads, 10M elements.}
    \label{fig:intro-query-accuracy}
\end{figure}

In Section~\ref{sec:eval} we empirically evaluate \mysketch. We show an update speedup of $12$x and a query speedup of $30$x over the sequential sketch, both with linear speedup. We compare \mysketch to FCDS, which is the state-of-the-art in concurrent sketches, and show that for FCDS to achieve similar performance it requires an order of magnitude larger buffers that \mysketch, reducing query freshness tenfold.

In the supplementary material we formally define the system model and present formal correctness proofs.





\section {Background}
\label{sec:background}
\subsection{Problem Definition}
Given a stream $A=x_1,x_2,\dots,x_n$ with $n$ elements,
the \emph{rank} of some $x$ (not necessarily in $A$) is the number of elements smaller than $x$ in $A$, denoted $R(A,x)$. For any $0 \leq \phi \leq 1$, the \emph{$\phi$ quantile} of $A$ is an element $x$ such that $R(A,x)=\lfloor \phi n \rfloor$.

A Quantiles sketch's API is as follows:
\begin{itemize}
\item \textbf{update(}$x$\textbf{)} process stream element $x$;
\item \textbf{query(}$\phi$\textbf{)} return an approximation of the $\phi$ quantile in the stream processed so far. 
\end{itemize}
A PAC Quantiles sketch with parameters $\epsilon, \delta$ returns element $x$ for query($\phi$) after n updates such that $R(A,x) \in \left[ (\phi-\epsilon)n,(\phi+\epsilon)n  \right]$, with probability at least $1-\delta$.

In an $r$-relaxed sketch for some $r\geq0$ every query returns an estimate of the $\phi$ quantile in a subset of the stream processed so far including all but at most $r$ stream elements~\cite{Henzinger_2013_Quantitative_Relaxation, Rinberg_2020_fast_sketches}.

\subsection{Sequential Implementation} \label{Section: sequential_imp}

The Quantiles sketch proposed by Agarwal et al.~\cite{mergeables_summaries} consists of a hierarchy of arrays, where each array summarizes a subset of the overall stream. The sketch is instantiated with a parameter $k$, which is a function of $(\epsilon,\delta)$. The first array, denoted level $0$, consists of at most $2k$ elements, and every subsequent array, in levels $1,2,\dots$, consists of either $0$ or $k$ elements at any given time.

Stream elements are processed in order of arrival, first entering level $0$, until it consists of $2k$ elements. Once this level is full, the sketch samples the array by sorting it and then selecting either the odd indices or the even ones with equal probability. The $k$ sampled elements are then propagated to the next level, and the rest are discarded. If the next level is full, i.e., consists of $k$ elements, then the sketch samples the union of both arrays by performing a merge sort, and once again retaining either the odd or even indices with equal probability. This propagation is repeated until an empty level is reached. Every level that is sampled during the propagation is emptied. Figure~\ref{fig: quantiles_sketch} depicts the processing of $4k$ elements.

\begin{figure}[h]
    \centering
    \includegraphics[width=\columnwidth]{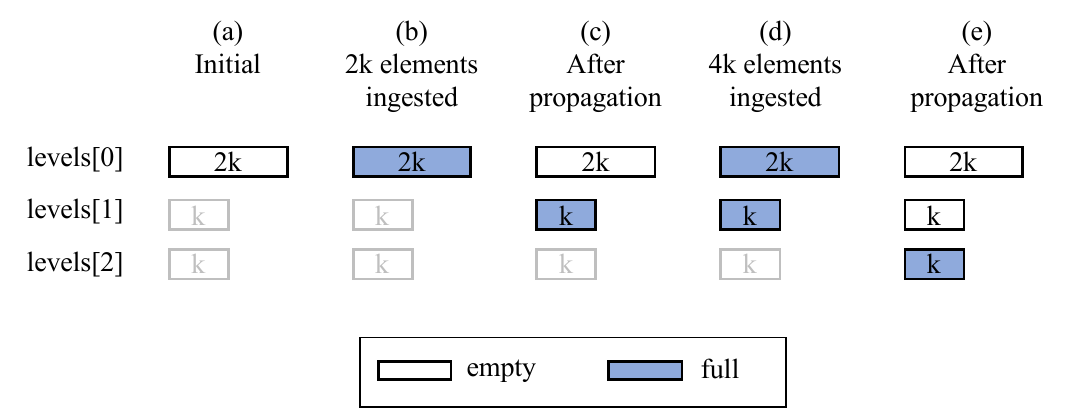}
    \caption{Quantiles sketch structure and propagation.}
    \label{fig: quantiles_sketch}
\end{figure}

Each element is associated with a \emph{weight}, which is the number of coin flips it has ``survived''. An element in an array on level $i$ has a weight of $2^i$, as it was sampled $i$ times. Thus, an element with a weight of $2^i$ represents $2^i$ elements in the processed stream.
For approximating the $\phi$ quantile, we construct a list of tuples, denoted \emph{samples}, containing all elements in the sketch and their associated weights. The list is then sorted by the elements' values. Denote by $W(x_i)$ the sum of weights up to element $x_i$ in the sorted list. The estimation of the $\phi$ quantile is an element $x_j$, such that $W(x_j) \leq \lfloor\phi n\rfloor$ and $W(x_{j+1}) > \lfloor\phi n\rfloor$.

\section{\mysketch}
\label{sec:Quancurrent}
We present \mysketch, an $r$-relaxed concurrent Quantiles sketch where $r$ depends on system parameters as discussed below. The algorithm uses $N$ update threads to ingest stream elements and allows an unbounded number of query threads. Queries are processed at any time during the sketch's construction. 
We consider a shared memory model that provides synchronization variables (atomics) and atomic operations to guarantee sequential consistency as in C++~\cite{Boehm_2008_cpp}. Everything that happened-before a write in one thread becomes visible to a thread that reads the written value. Also, there is a single total order in which all threads observe all modification in the same order. 
We use the following sequentially consistent atomic operations (which force a full fence): \emph{fetch-and-add (F\&A)}~\cite{x86-faa} and \emph{compare-and-swap (CAS)}~\cite{x86-cas}. 

In addition, we use a software-implemented higher-level primitive, \emph{double-compare-double-swap (DCAS)} which atomically updates two memory addresses as follows: DCAS($addr_1$: $old_1 \to new_1$, $addr_2$: $old_2 \to new_2$)
is given two memory addresses $addr_1$, $addr_2$, two corresponding expected values $old_1$, $old_2$, and two new values $new_1$, $new_2$ as arguments. It atomically sets $addr_1$ to $new_1$ and $addr_2$ to $new_2$ only if both addresses match their expected values, i.e., the value at $addr_1$ equals $old_1$ and the value at $addr_2$ equals $old_2$. DCAS also provides wait-free DCAS\_READ primitive, which can read fields that are concurrently modified by a DCAS. DCAS can be efficiently implemented using single-word CAS~\cite{Harris2002practical,guerraoui2020efficient}.

In Section~\ref{Section: data_organization}, we present the data structures used by Quancurrent. Section~\ref{Section: concurrent_algorithm} presents the update operation, and Section~\ref{ssec:query} presents the query. The formal correctness proof is deferred to the supplementary material. 

\subsection{Data Structures} \label{Section: data_organization}
\mysketch's data structures are described in Algorithm~\ref{alg: data_organization} and depicted in Figure~\ref{fig: data_structures}.
Similarly to the sequential Quantiles sketch, \mysketch is organized as a hierarchy of arrays called \emph{levels}. Each level can be \emph{empty}, \emph{full}, or in \emph{propagation}. The variable \emph{tritmap} maintains the states of all levels. Tritmap is an unsigned integer, interpreted as an array of trits (trinary digits). The trit $tritmap[i]$ describes level $i$'s state: if $tritmap[i]$ is $0$, level $i$ contains $0$ or $2k$ ignored elements and is considered to be empty. If $tritmap[i]$ is $1$, level $i$ contains $k$ elements and is deemed full, and if it is $2$, level $i$ contains $2k$ elements and is associated with the propagation state. Each thread has a local buffer of size $b$, $\mathit{localBuf[b]}$. Before ingested into the sketch's levels, stream elements are buffered in threads local buffers and then moved to a processing unit called $\mathit{Gather\&Sort}$. The $\mathit{Gather\&Sort}$ object has two $2k$-sized shared buffers, $\mathit{G\&SBuffer}[2]$, each with its own $\mathit{index}$ specifying the current location, as depicted in Figure~\ref{fig: gather_and_sort}. 

The query mechanism of \mysketch includes taking an atomic snapshot of the levels. Query threads cache the snapshot and the tritmap that represents it in local variables, 
$\mathit{snapshot}$ and $\mathit{myTrit}$, respectively. As the snapshot reflects only the sketch's levels and not G\&SBuffers or the thread's local buffers, Quancurrent is ($4kS+(N-S)b$)-relaxed Quantiles sketch where $S$ is the number of NUMA nodes. 

\begin{algorithm}[]
\caption{\mysketch data structures} \label{alg: data_organization}
\KwSty{Parameters and constants:}\;
\Indp
    \myvar{MAX\_LEVEL} \;
    \myvar{k} \Comment*[r]{sketch level size}
    \myvar{b}  \Comment*[r]{local buffer size}
    \myvar{S} \Comment*[r]{\#NUMA nodes}
\Indm

\KwSty{Shared objects:} \;
\Indp
    \myvar{tritmap} $\gets 0$ \;
    \myvar{levels}[\myvar{MAX\_LEVEL}]\;
\Indm

\KwSty{NUMA-local objects:} \; \Comment*[r]{shared among threads on the same node}
\Indp
    \myvar{G\&SBuffer}[2][2\myvar{k}] \;
    \myvar{index}[2] $\gets \{0,0\}$ \;
\Indm

\KwSty{Thread local objects:} \;
\Indp
    \myvar{localBuf}[\myvar{b}] \;
    \myvar{myTrit} \Comment*[r]{used by query}
    \myvar{snapshot} \Comment*[r]{used by query}
\Indm
\end{algorithm}

\subsection{Update} \label{Section: concurrent_algorithm} 
The ingestion of stream elements occurs in three stages: 
(1) \emph{gather and sort}, 
(2) \emph{batch update}, and 
(3) \emph{propagate level}. 
In stage (1), stream elements are buffered and sorted into batches of $2k$ through a $\mathit{Gather\&Sort}$ object. Each NUMA node has its designated $\mathit{Gather\&Sort}$ object, which is accessed by NUMA-local threads. 
Stage (2) executes a batch update of $2k$ elements from the $\mathit{Gather\&Sort}$ object to $levels[0]$. 
Finally, in stage (3), $levels[0]$ is propagated up the levels of the hierarchy.

In the first stage, threads first process stream elements into a thread-local buffer of size $b$. Once the buffer is full, it is sorted and the thread reserves $b$ slots on a shared buffer in its node's Gather\&Sort unit. It then begins to move the local buffer's content to the shared buffer. The shared Gather\&Sort buffer contains $2k$ elements, and its propagation (during Stage 2) is not synchronized with the insertion of elements. Thus, some reserved slots might still contain old values, (which have already been propagated), instead of new ones. As the batch is a sample of the original stream, we can accept the possible loss of information in order to improve performances. Below, we show that the sampling bias this introduces is negligible.

The pseudo-code for the first stage is presented in Algorithm~\ref{alg: gather}. To insert its elements to the shared buffer, a thread tries to reserve $b$ places in one of the shared buffers using F\&A (Line~\ref{Line: fetch_idx}). If the index does not overflow, the thread copies its local buffer to the reserved slots (Line~\ref{Line: copy_local_buffer}). We refer to the thread that fills the last $b$ locations in a G\&SBuffer as the \emph{owner} of the current batch. The batch owner creates a local sorted copy of the shared buffer and begins its propagation (Lines~\ref{Line:create_copy}-\ref{Line:call_batchUpdate}).

Note that the local buffer is not atomically moved into the shared buffer (Line \ref{Line: copy_local_buffer} is a loop). Thus, the owner might begin a propagation before another thread has finished moving its elements to the shared buffer. In this case, the old elements already contained within the G\&SBuffer are taken instead. Furthermore, upon moving its elements later, the writer thread might overwrite more recent elements. In other words, during this stage, stream elements may be duplicated and new elements may be dropped. We call both of these occurrences \emph{holes}, and analyze their implications in Section~\ref{ssec:holes-analysis}.

\begin{algorithm}[]
\caption{Stage 1: gather and sort} \label{alg: gather}
\SetKwFunction{update}{update}
\Procedure{\update{x}}{
add $\mathit{x}$ to \myvar{localBuf} \Comment*[r]{thread-local} \label{Line: update_linearization}
    \lIf{$\neg \mathit{localBuf}$.full()}{\KwRet{}}
    sort \myvar{localBuf} \;
    \myvar{i} $\gets 0$ \;
    \Repeat(\Comment*[f]{insert to Gather\&Sort unit}){$\mathit{true}$}{
        \myvar{idx} $\gets$ \myvar{index}[\myvar{i}].F\&A(\myvar{b}) \; \label{Line: fetch_idx}
        \uIf(\Comment*[f]{space available}){$\mathit{idx}< 2k$}{
            move \myvar{localBuf} to \myvar{G\&SBuffer}[\myvar{i}][$\mathit{idx}, \dots, \mathit{idx}+b$] \; \label{Line: copy_local_buffer}
            \uIf(\Comment*[f]{owner, filled buffer}){$\var{idx}+b=2k$}{
                $\mathit{myCopy} \gets \text{sorted copy of } \mathit{G\&SBuffer}[\mathit{i}]$ \; \label{Line:create_copy}
                batchUpdate(\myvar{i},\myvar{myCopy}) \; \label{Line:call_batchUpdate}
            }
            \KwRet{}
        }
        \myvar{i} $\gets$ $\neg$\myvar{i}\;
    }   
}
\end{algorithm}

\begin{figure}[]
    \begin{subfigure}[b]{\linewidth}
    \centering
    \includegraphics[width=0.4\linewidth]{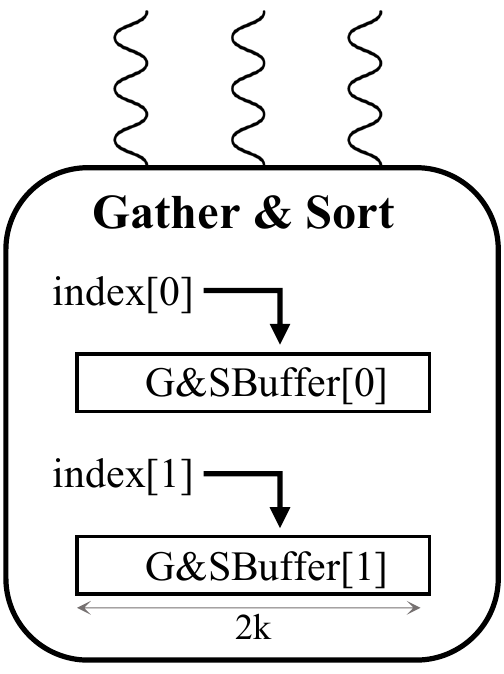}
    \caption{$\mathit{Gather\&Sort}$ object.}
    \label{fig: gather_and_sort}
    \end{subfigure}
    \par\medskip
    \begin{subfigure}[b]{\linewidth}
    \centering
    \includegraphics[width=0.5\linewidth]{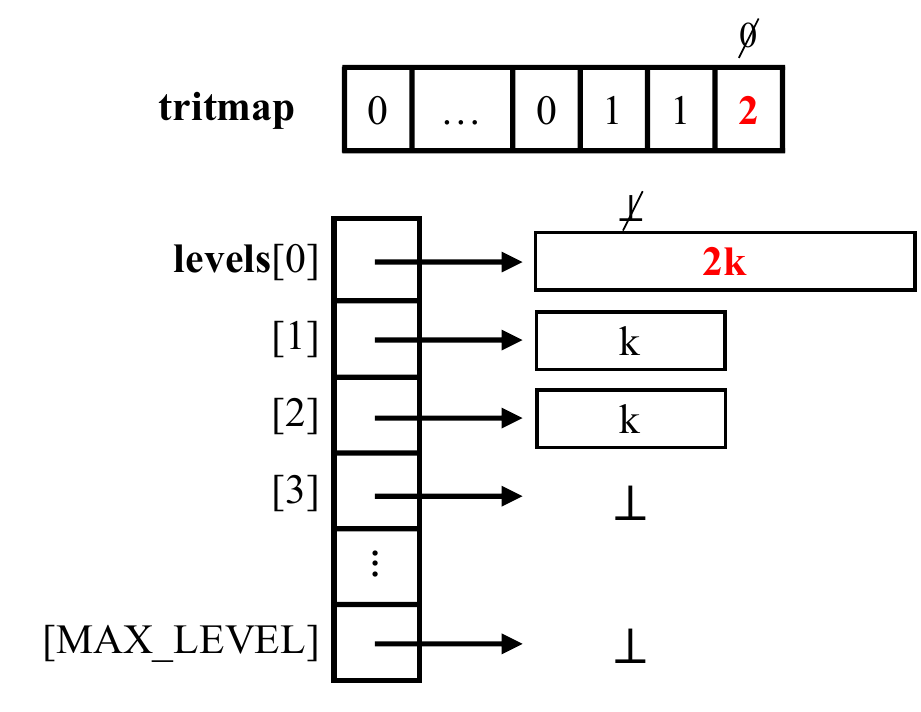}
    \caption{Batch update into $\mathit{levels}[0]$.}
    \label{fig: batch_update}
    \end{subfigure}%
    
    \caption{Quancurrent's data structures.}
    \label{fig: data_structures}
    
\end{figure}

In the second stage, the owner inserts its local sorted copy of the shared buffer into level $0$ using a DCAS. The batch of $2k$ elements is only inserted when level $0$ is empty, reflected by the first digit of the tritmap being $0$. We use DCAS to atomically update both \emph{levels}[0] to point to the new sorted batch and \emph{tritmap} to indicate an ongoing batch update (reflected by setting $\mathit{tritmap[0]}$ to 2). The DCAS might fail if other owner threads are trying to insert their batches or propagate them. The owner keeps trying to insert its batch into the sketch's first level until a DCAS succeeds, and then resets the index of the G\&SBuffer to allow other threads to ingest new stream elements. 
The pseudo-code for the second stage is presented in Algorithm~\ref{alg: batch_update}, and an example is depicted in Figure~\ref{fig: batch_update}.

\begin{algorithm}[]
\caption{Stage 2: batch update} \label{alg: batch_update}
\SetKwFunction{batchUpdate}{batchUpdate}
\Procedure{\batchUpdate{\myvar{i},\myvar{base\_copy}}}{
    \lWhile(\label{Line:insert_batch}){$\neg$DCAS(\myvar{levels}[0]: $\bot \to$ \myvar{base\_copy}, \myvar{tritmap}[0]: 0 $\to$ 2 )}{\{ \}}
    \myvar{index}[\myvar{i}] $\gets 0$ \;
    propagate(0)\;
}
\end{algorithm}

In the beginning of the third stage, level 0 points to a new sorted copy of a \emph{G\&SBuffer} array and \emph{tritmap}[0]=2. During this stage, the owner thread propagates the newly inserted elements up the levels hierarchy iteratively, level by level from level 0 until an empty level is reached. The pseudo-code for the propagation stage is presented in Algorithm~\ref{alg: propagate}. On each call to \emph{propagate}, level $l$ is propagated to level $l+1$, assuming that level $l$ contains $2k$ sorted elements and $\mathit{tritmap}[l]=2$. If $\mathit{tritmap}[l+1]=2$, the owner thread is blocked by another propagation from $l+1$ to $l+2$ and it waits until $\mathit{tritmap}[l+1]$ is either a $0$ or $1$. The owner thread samples $k$ elements from level $l$ and retains the odd or even elements with equal probability (Line~\ref{Line: choose_rand}). 
If $\mathit{tritmap}[l+1]$ is $1$, then level $l+1$ contains $k$ elements. The sampled elements are merged with level $l$+1 elements into a new $2k$-sized sorted array (Line~\ref{Line: next_full}). We then (in Line~\ref{Line:next_full_DCAS}) continuously try, using DCAS, to update \emph{levels}[$l$+1] to point to the merged array and atomically update \emph{tritmap} such that \emph{tritmap}[$l$] $\gets 0$, reflecting level $l$ is available, and \emph{tritmap}[$l$+1] $\gets 2$, reflecting that level $l$+1 contains $2k$ elements. After a successful DCAS, we clear level $l$ (set it to $\bot$) and proceed to propagate the next level (Line~\ref{Line:propagate_next}). 
If $\mathit{tritmap}[l+1]$ is $0$, then level $l+1$ is empty. We use DCAS (Line~\ref{Line:next_empty_DCAS}) to update $\mathit{levels}[l+1]$ to point to the sampled elements and atomically update \emph{tritmap} so that \emph{tritmap}[$l$] becomes $0$, and \emph{tritmap}[$l$+1] becomes $1$ (containing $k$ elements). After a successful DCAS, we clear level $l$ (set it to $\bot$) and end the current propagation.

Propagations of different batches may occur concurrently, i.e., level propagation of levels $l$ and $l'$ can be performed in parallel. Figure~\ref{fig: propagate} depicts an example of concurrent propagation of two batches.

\begin{algorithm}[h]
\caption{Stage 3: Propagation of level $l$} \label{alg: propagate}
\SetKwFunction{propagate}{propagate}
\Procedure{\propagate{\myvar{l}}}{
    \lIf(){\myvar{l} $\geq$ \myvar{MAX\_LEVEL}}{\KwRet{}}
    \Comment*[l]{choose odd or even indexed elements randomly}
    \myvar{newLevel} $\gets$ sampleOddOrEven(\myvar{levels}[\myvar{l}]) \; \label{Line: choose_rand}
    \uIf(\Comment*[f]{next level is full}){\myvar{tritmap}[\myvar{l}$+1$] $ =1$}{
         \myvar{newLevel} $\gets$ merge(\myvar{newLevel}, \myvar{levels}[\myvar{l}$+1$]) \label{Line: next_full}\;
         \lWhile(\label{Line:next_full_DCAS}){$\neg$DCAS(\myvar{levels}[\myvar{l}$+1$]: \myvar{levels}[\myvar{l}$+1$] $\to$ \myvar{newLevel},  \myvar{tritmap}[\myvar{l}, \myvar{l}$+1$]: $ [2,1] \to [0,2]$)}{\{ \}} 
        \myvar{levels}[\myvar{l}] $\gets \bot$ \Comment*[r]{clear level} 
        \KwRet{propagate(\myvar{l}+1)} \label{Line:propagate_next}
    }
    \Comment*[l]{\myvar{tritmap}[\myvar{l}+1] is 0 or 2}
    \lWhile(\label{Line:next_empty_DCAS}){$\neg$DCAS(\myvar{levels}[\myvar{l}$+1$]: $ \bot \to$ \myvar{newLevel}, \myvar{tritmap}[\myvar{l}, \myvar{l}$+1$]: $ [2,0] \to [0,1]$)}{\{ \}} 
        \myvar{levels}[\myvar{l}] $\gets \bot$ \Comment*[r]{clear level} 
}
\end{algorithm}

\begin{figure*}[h]
\includegraphics[width=0.8\textwidth,trim={0 0cm 0 0},clip]{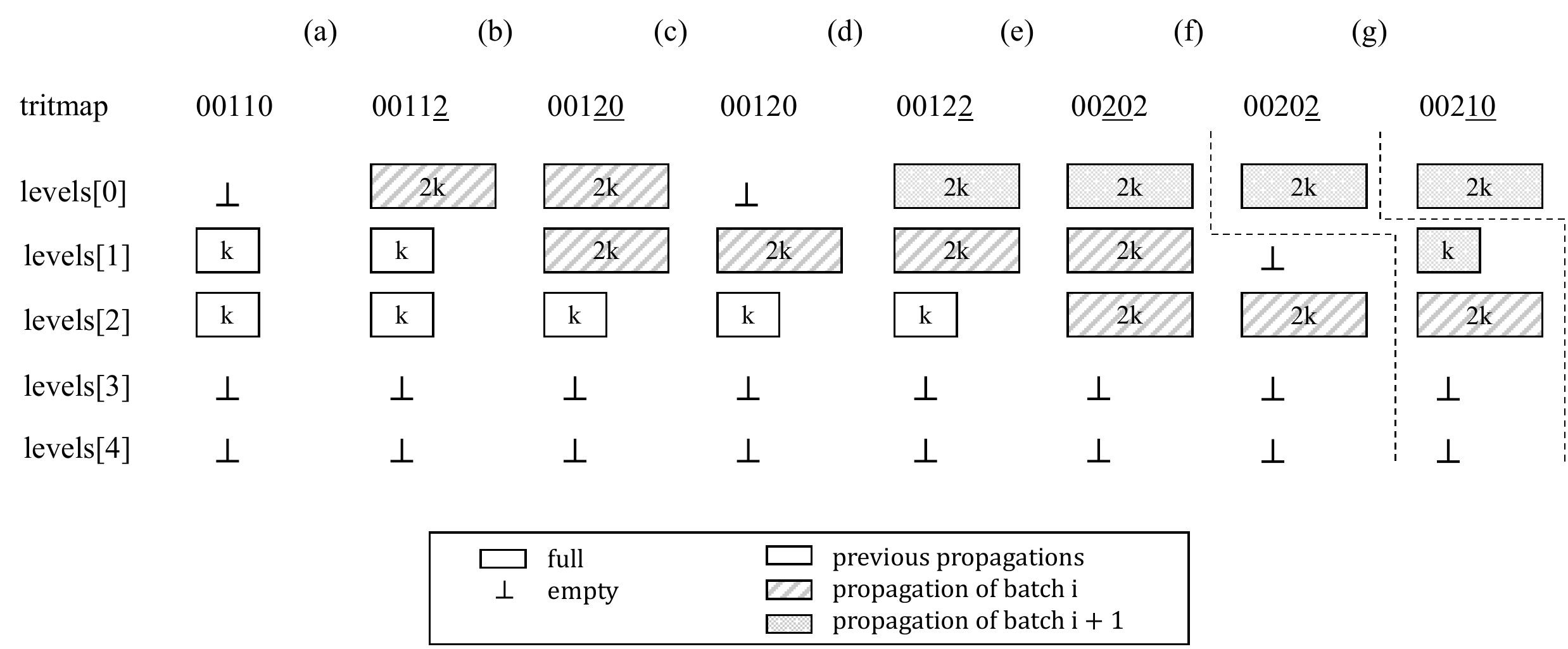}
\centering
\captionsetup{justification=centering}
\caption[\mysketch~\xspace propagation.]{\mysketch~\xspace propagation. \par \small \textbf{(a)} The owner of batch $i$, owner($i$), inserts batch $i$ to level $0$ and atomically updates $tritmap[0]$ to $2$. 
\textbf{(b)} owner($i$) merges level $0$ with level $1$ and changes $tritmap[1,0]$ from $[1,2]$ to $[2,0]$. 
\textbf{(c)} owner($i$) clears level $0$. 
\textbf{(d)} owner($i+1$) inserts its batch to level $0$ and atomically updates $tritmap[0]$ to $2$. 
\textbf{(e)} owner($i$) merges level $1$ with level $2$, and sets $tritmap[2,1]$ to $[2,0]$. Batch $i+1$ is still blocked because level $1$ has not been cleared yet. 
\textbf{(f)} owner($i$) clears level $1$. 
\textbf{(g)} Now owner($i+1$) successfully merges level $0$ with the empty level $1$, and sets $tritmap[1,0]$ to $[1,0]$. 
}
\label{fig: propagate}
\end{figure*}

\subsection{Query}
\label{ssec:query}

Queries are performed by an unbounded number of query threads. A query returns an approximation based on a subset of the stream processed so far including all elements whose propagation into the levels array begun before the query was invoked. The query is served from an atomic snapshot of the levels array. The pseudo-code is presented in Algorithm~\ref{alg: sl_query}. Instead of collecting a new snapshot for each query, we cache the snapshot so that queries may be serviced from this cache, as long as the snapshot isn't too stale. The snapshot and the tritmap value that represents it are cached in local variables, $\mathit{snapshot}$ and $\mathit{myTrit}$, respectively. Query freshness is controlled by the parameter $\rho$, which bounds the ratio between the current stream size and the cached stream size. As long as this threshold is not exceeded, the cached snapshot may be returned (Lines~\ref{Line:check_rho}-\ref{Line:query_from_cache}). Otherwise, a new snapshot is taken and cached. 

The snapshot is obtained by first reading the tritmap, then reading the levels from $0$ to $MAX\_LEVEL$, and then reading the tritmap again. If both reads of the tritmap represent the same stream size then they represent the same stream. We can use the levels read to reconstruct some state that represents this stream. The process is repeated until two such tritmap values are read. For example, focusing on the last two phases of the propagation in Figure~\ref{fig: propagate}, lets assume a query thread $T_q$ reads $tm1=00202$, then reads the levels from $\mathit{levels[0]}$ to $\mathit{levels[4]}$ as depicted in Figure~\ref{fig: propagate} (between the dashed lines), and then read $tm2=00210$. The two tritmap reads represent the same stream of size $10k$, thus a snapshot representing the same stream can be constructed from the levels read. The pseudo-code for calculating the stream size is presented in Algorithm~\ref{alg:tritmap}. Each level is read atomically as the levels arrays are immutable and replaced by pointer swings. The snapshot is a subset of the levels summarizing the stream. To construct the snapshot, the collected levels are iterated over, in reversed order, from $MAX\_LEVEL$ to $0$, and level $i$ is added to the snapshot only if the total collected stream size (including level $i$) is less than or equal to the stream size represented by the tritmap (Line~\ref{Line:add_to_snapshot}). Back to our last example, the size of each level collected by $T_q$ is ${2k,k,2k,0,0}$ (in descending order). As explained, to construct the snapshot, we go over the collected levels from $\mathit{snapLevels[4]}$ to $\mathit{snapLevels[0]}$. By reading $snapLevels[1]$, the total stream size represented by snapshot is $0+0+4\cdot2k+2\cdot k = 10k$. As the stream size represented by $tm1$ and $tm2$ is $10k$, the construction of the snapshot is done and all elements of the processed stream are represented exactly once. The tritmap $\mathit{myTrit}$ maintains the total size of the collected stream and each trit describes the state of a collected level. If level i was collected to the snapshot, the value of $\mathit{myTrit[i]}$ is the size of level i divided by $k$ (Line~\ref{Line:update_myTrit}).

As levels propagate from lowest to highest, reading the levels in the same direction ensures that no element would be missed but may cause elements to be represented more than once. Building the snapshot bottom-up ensures that each element will be accounted once. 
In other words, reading the levels from lowest to highest and building the snapshot highest to lowest ensures that an atomic snapshot is collected, as proven in the Supplementary material.


\begin{algorithm}[]
\caption{Query} \label{alg: sl_query}
\SetKwFunction{Query}{Query}
\Procedure{\Query{$\phi$}}{
    \myvar{tm1} $\gets$ tritmap \; \label{Line:read_current_stream}
    \uIf{$\frac{tm1.streamSize()}{myTrit.streamSize()}$ $\leq$  $\rho$}{ \label{Line:check_rho}
        \KwRet{snapshot.query($\phi$)} \; \label{Line:query_from_cache}
    }
    \Repeat(){\myvar{tm1}.streamSize() $\ne$ \myvar{tm2}.streamSize() }  { \label{Line:start_collect_levels}  
         \myvar{tm1} $\gets$ tritmap \;
         \myvar{snapLevels} $\gets$ read \myvar{levels} $0$ to \myvar{MAX\_LEVEL} \; \label{Line: read_snap}
         \myvar{tm2} $\gets$ \myvar{tritmap} \; \label{Line: second-collect}
    } \label{Line:query_linearization}
    \myvar{myTrit} $\gets 0$ \; \label{line: query_estimate_start}
    \myvar{snapshot} $\gets$ empty \myvar{snapshot} \;
    \For{\myvar{i} $\gets$ \myvar{MAX\_LEVEL} \KwTo  $0$}{ \label{Line:start_collect_snapshot}
        \myvar{weight} $\gets 2^i $ \;
        \uIf(\label{Line:add_to_snapshot}){\myvar{snapLevels}[\myvar{i}].size()$\cdot$\myvar{weight}+ \myvar{myTrit}.streamSize()$\leq$\myvar{tm1}.streamSize()}{
            add \myvar{snapLevels}$[i]$ to \myvar{snapshot} \; \label{Line:add_level_to_snapshot}
            \myvar{myTrit}[\myvar{i}] $\gets$ \myvar{snapLevels}[\myvar{i}].size()$ / k$ \; \label{Line:update_myTrit} \label{line: query_estimate_end}
            \lIf{\myvar{myTrit}.streamSize()=\myvar{tm1}.streamSize()} { \label{Line:snapshot_done}
                break
            }
        }
    }
    \KwRet{\myvar{snapshot}.query($\phi$)} \label{Line:query_result}
}
\end{algorithm}

\begin{algorithm}[]
\caption{Tritmap} \label{alg:tritmap}
\SetKwFunction{streamSize}{streamSize}
    \Procedure{\streamSize{}} {
        \myvar{curr\_stream} $\gets$ 0 \;
        \For{\myvar{i} $\gets$ 0 \KwTo \myvar{MAX\_LEVEL}}{
            \myvar{weight} $\gets 2^i $ \;
            \uIf{\myvar{tritmap}[\myvar{i}] $=1$}{
                \myvar{curr\_stream} $\gets$ \myvar{curr\_stream} $+$ \myvar{weight}$\cdot k$
            }
            \uElseIf{\myvar{tritmap}[\myvar{i}] $=2$}{
                     \myvar{curr\_stream} $\gets$ \myvar{curr\_stream} $+$ \myvar{weight}$\cdot 2k$
            }
        }
        \KwRet{\myvar{curr\_stream}}
    }
\end{algorithm}


\section{Analysis}
\label{sec:analysis}
In Section~\ref{ssec:holes-analysis} we analyze the expected number of holes, and in Section~\ref{ssec:error-analysis} we analyze \mysketch's error.

\subsection{Holes Analysis}
\label{ssec:holes-analysis}
Because the update operation moves elements from the thread's local buffer to a shared buffer non-atomically, holes may occur when the owner thread reads older elements that were written to the shared buffer in a previous window. The missed (delayed) writes may later overwrote newer writes. Together, for each hole, an old value is duplicated and a new value is dropped.

We analyze the expected number of holes under the assumption of a \emph{uniform stochastic scheduler}~\cite{alistarh2016lock}, which schedules each thread with a uniform probability in every step. That is, at each point in the execution, the probability for each thread to take the next step is $\frac{1}{N}$.

Denote by $H$ the total number of holes in some batch of $2k$ elements. G\&SBuffer's array is divided into $\frac{2k}{b}$ regions, each consisting of $b$ slots populated by the same thread. Denote by $H_1,\dots, H_{\frac{2k}{b}}$ the number of holes in regions $1, \dots, \frac{2k}{b}$, respectively.

The slots in region $j$ are written to by the thread that successfully increments the shared index from $(j-1)b$ to $jb$. We refer to this thread as $T_j$. Note that multiple regions may have the same writing thread. The shared buffer's owner, $T_O$, is $T_{\frac{2k}{b}}$. To initiate a batch update, $T_O$ creates a local copy of one G\&SBuffer by iteratively reading the array. A hole is read in some region $j$ if $T_O$ reads some index $i+1$ in this region before the writer thread $T_j$ writes to the corresponding index in the same region.

\paragraph*{Analysis of $\boldsymbol{H_j}$.} When $T_O$ increments the index from $2k-b$ to $2k$, $T_j$ may have completed any number of writes between $0$ and $b$ to region $j$. We first consider the case that $T_j$ hasn't completed any writes. In this case, for a hole to be read in slot $i+1$ of region $j$, $T_O$'s read of slot $i+1$ must overtake $T_j$'s write of the same slot. To this end, $T_O$ must write $b$ values (from its own local buffer), read $(j-1)b$ values from the first $j-1$ regions and then read values from slots $1,\dots,i+1$ in this region before $T_j$ takes $i+1$ steps. The probability that $T_O$ reads a hole for the first time in this region in slot $i+1$ is:
\begin{align*}
\pi_{i,j} \triangleq P[\text{hole in slot }i+1  &\mid \text{no hole in slots } 1\dots i] \\
                                            & \cdot P[\text{no hole in slots } 1\dots i].
\end{align*}
For a hole to be read in slot $i+1$ of region $j$, $T_O$ must take $b+(j-1)b+i+1$ steps while $T_j$ takes at most $i$ steps, with $T_O$'s read of slot $i+1$ being last.
But if $T_j$ takes fewer than $i$ steps, a hole is necessarily read earlier than slot $i+1$. Therefore, we can bound $\pi_{i,j}$ by considering the probability that $T_j$ takes exactly $i$ steps while $T_O$ takes $b+(j-1)b+i$ steps, and then $T_O$ takes a step. Ignoring steps of other threads, each of $T_j$ and $T_O$ has a probability of $\frac{1}{2}$ to take a step before the other. Therefore,
\[\pi_{i,j} \leq \left(\frac{1}{2}\right)^{jb + 2i +1} {{jb+2i} \choose i}.\]

Note that this includes schedules in which $T_O$ reads holes in previous slots in the same region, therefore it is an upper bound.
Given that $T_j$ has not yet written in region $j$, the probability, $p_j$, that $T_O$ reads at least $1$ hole in region $j$ is bounded as follows:
\begin{align*}
    p_j &\leq \sum_{i=0}^{b-1} \pi_{i,j}
\end{align*}


If $T_j$ has completed writes to region $j$, the probability that $T_O$ reads holes is even lower.
Therefore, the probability that $H_j \geq 1$ is bounded from above by $p_j$.
Using this, we bound the expected total number of holes in region $j$:
\[E\left[H_j\right]= P(H_j=0)\cdot 0 + P(H_j=1)\cdot 1 + \dots + P(H_j=b)\cdot b.\]
$T_O$ can read at most $b$ holes, therefore,
\begin{flalign*}
E\left[H_j\right] &< b\cdot \left(P(H_j=1) + \dots + P(H_j=b) \right) \\
                    &= b\cdot P(H_j\geq1) < b\cdot p_j.
\end{flalign*}



Using the linearity of expectation, we bound the expected number of holes in a batch:
\[E\left[H\right]= E\left[H_1\right] + E\left[H_2\right] + \dots + E\left[H_{\frac{2k}{b}}\right]. \]

In the supplementary material, we prove that
\begin{align*}
   &\forall j\geq1, b\in \mathds{N}, \; E[H_{j+1}] \leq 0.5 \cdot E[H_j] \\ 
    &\forall b \in \mathds{N}, \;  E[H_1] \leq 1.4 
\end{align*}
Together, this implies that  $E[H] \leq 2.8$ for all $b \in \mathds{N}$.

\subsection{Error Analysis}
\label{ssec:error-analysis}
The source of \mysketch's estimation error is twofold: (1) the error induced by sub-sampling the stream, and (2) the additional error induced by concurrency. For the former, we leverage the existing literature on analysis of sequential sketches. We analyze the latter. 
As the expected number of holes is fairly small and the holes are random, we disregard their effect on the error analysis. 

First, our buffering mechanism induces a relaxation. Let $S$ be the number of NUMA nodes. Recall that each NUMA node has a Gather\&Sort object that contains two buffers of size $2k$. In addition, each of the $N$ update threads has a local buffer. When the G\&SBuffer is full, the local buffer of the owner is empty so at most $N-S$ threads lcally buffered elements. Therefore, the buffering relaxation $r$ is $4kS+(N-S)b$.

Rinberg et al.~\cite{Rinberg_2020_fast_sketches} show that for a query of a $\phi$-quantile, an $r$-relaxation of a Quantiles sketch with parameters $\epsilon_c$ and $\delta_c$, returns an element whose rank is in the range $[(\phi-\epsilon_r)n,(\phi+\epsilon_r)n]$ with probability at least $1-\delta_c$, for $\epsilon_r=\epsilon_c+\frac{r}{h}(1-\epsilon_c)$.

On top of this relaxation, our cache mechanism induces further staleness. Here, the staleness depends on $\rho$. Let $n_{old}$ be the stream size of the cached snapshot, and let $n_{new}$ be the current stream size. If $n_{new} / n_{old} \leq \rho$ then the query is answered from the cached snapshot.
Denote $\rho \triangleq 1+\epsilon'$ for some $\epsilon' \geq 0 $. The element returned by the cached snapshot is in the range:
\[\left[\left(\phi-\epsilon_r\right)n_{old},\left(\phi+\epsilon_r\right)n_{old}\right]\]
As $n_{old} \leq n_{new}$, then, 
\begin{flalign*}
\left(\phi+\epsilon_r\right)n_{old} \leq \left(\phi+\epsilon_r\right)n_{new} \leq \left(\phi+\left(\epsilon'+\epsilon_r\right)\right)n_{new} &&
\end{flalign*}
On the other hand, 
\begin{flalign*}
\left(\phi-\epsilon_r\right)n_{old} \geq &\left(\phi-\epsilon_r\right)\frac{n_{new}}{\rho}= &&\\
&\left(\frac{\phi}{1+\epsilon'}-\frac{\epsilon_r}{1+\epsilon'}\right)n_{new}= &&\\
&\left(\phi-\frac{\phi\epsilon'}{1+\epsilon'}-\frac{\epsilon_r}{1+\epsilon'}\right)n_{new} \geq &&\\
&\left(\phi-\left(\epsilon'+\epsilon_r\right)\right)n_{new}
\end{flalign*}
Because $\phi \leq 1$ and $\epsilon' \geq 0$ then, $\frac{\phi\epsilon'}{1+\epsilon'} \leq \frac{\epsilon'}{1+\epsilon'} \leq 1$. \\
Therefore, the query returns a value within the range
\[\left[\left(\phi-\epsilon \right)n,\left(\phi+\epsilon \right)n\right]\] for $\epsilon \triangleq \epsilon_r + \epsilon'$.

\section{Evaluation}
\label{sec:eval}
\graphicspath{{../images/graphs/accuracy/} {../images/graphs/parameters/} {../images/graphs/throughput/}}

In this section we measure \mysketch's throughput and estimation accuracy. 
Section~\ref{ssec:setup} presents the experiment setup and methodology.
Section~\ref{ssec:tput} presents throughput measurements and discusses scalability.
Section~\ref{ssec:params} experiments with different parameter setting, examining how performance is affected by query freshness. 
Section~\ref{ssec:accuracy} presents an accuracy of estimation analysis.
Finally, Section~\ref{ssec:compare} compares \mysketch\ to the state-of-the-art. 

\subsection{Setup and Methodology}
\label{ssec:setup}
We implement \mysketch in C++. Our memory management system is based on IBR~\cite{Haosen_2018_IBR}, an interval-based approach to memory reclamation for concurrent data structures. 
The experiments were run on a NUMA system with four Intel Xeon E5-4650 processors, each with 8 cores, for a total of 32 threads (with hyper-threading disabled).

Each thread was pinned to a NUMA node, and nodes were first filled before overflowing to other NUMA nodes, i.e., $8$ threads use only a single node, while $9$ use two nodes with $8$ threads on one and $1$ on the second. The default memory allocation policy is local allocation, except for \mysketch's shared pointers. Each Gather\&Sort unit is allocated on a different NUMA node and threads update the G\&SBuffers allocated on the node they belong to.
The stream is drawn from a uniform distribution, unless stated otherwise. Each data point is an average of 15 runs, to minimize measurement noise.

\subsection{Throughput Scalability}
\label{ssec:tput} 
We measured \mysketch's throughput in three workloads:
(1) update-only, (2) query-only, and (3) mixed update-query. In the update-only workload, we update \mysketch with a stream of 10M elements and measure the time it takes to feed the sketch. For the other two workloads, we pre-fill the sketch with a stream of 10M elements and then execute the workload (10M queries only or queries and 10M updates) and measure performance. Figure~\ref{fig:throughput} shows \mysketch's throughput in those workloads with $k = 4096$ and $b = 16$,

\begin{figure*}[ht!] 
\centering
    \begin{subfigure}[]{\columnwidth}
        \centering
        \includegraphics[width=0.7\columnwidth,trim={0cm 0cm 1.9cm 2.6cm},clip] {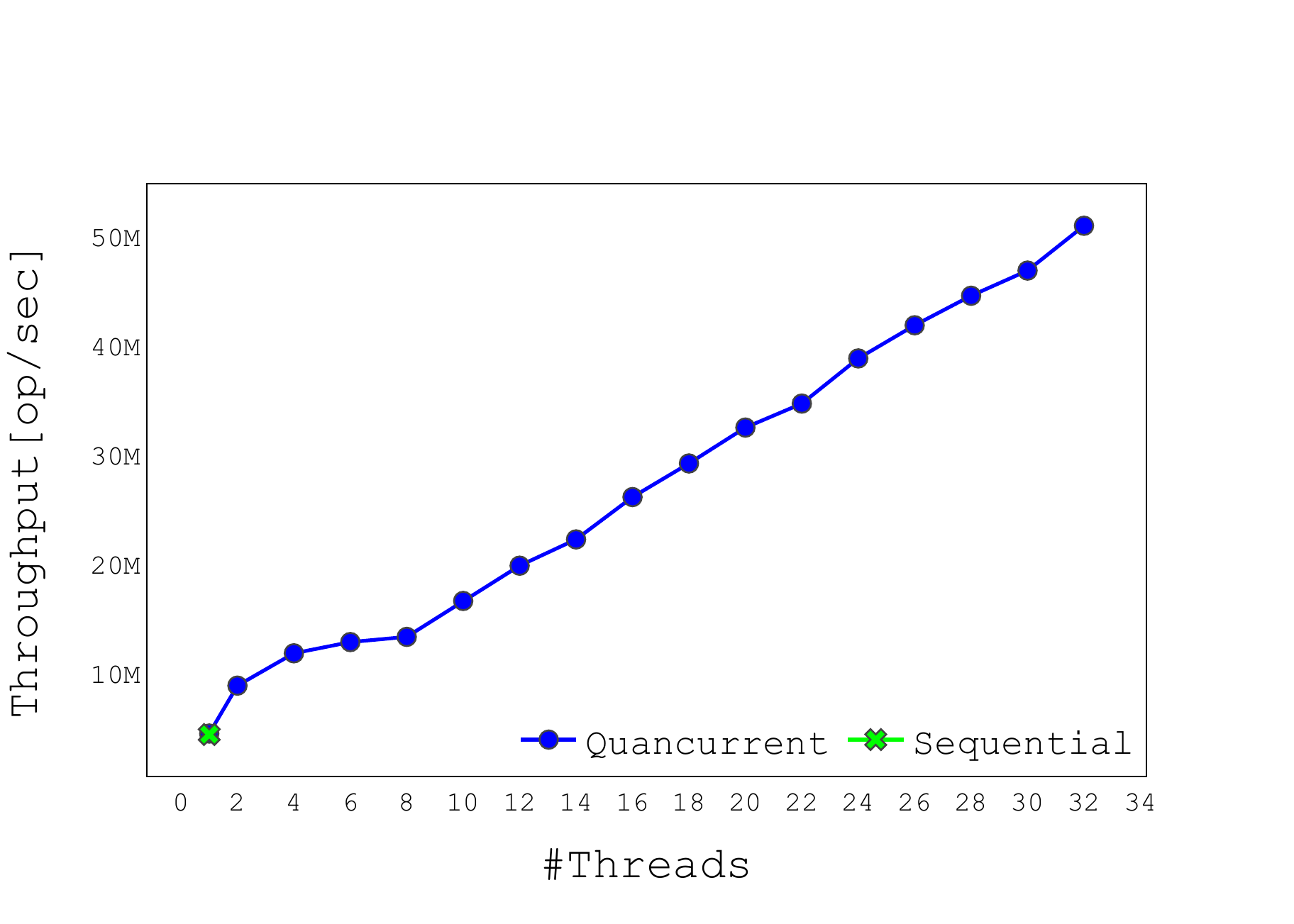}
        \caption{Update-only, 10M elements.}
        \label{fig:update_only_speedup}
    \end{subfigure}
    \hspace{-5em}
    \begin{subfigure}[]{\columnwidth}
        \centering
        \includegraphics[width=0.7\columnwidth,trim={0cm 0cm 1.9cm 2.6cm},clip] {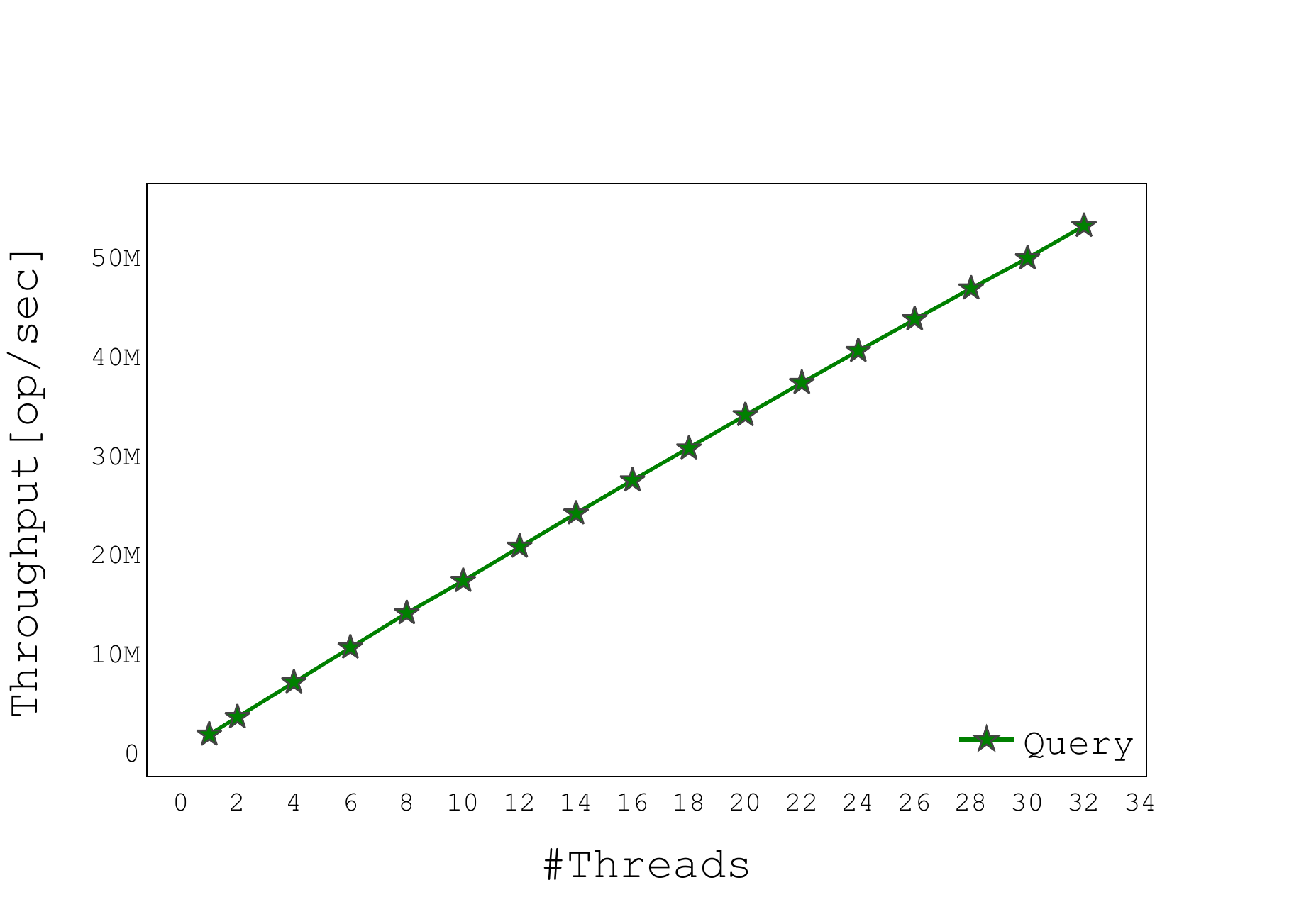}
        \caption{Query-only, 10M elements prefilled, 10M queries.}
        \label{fig:query_only_throughput}
    \end{subfigure}
\vfill
    \begin{subfigure}[]{\textwidth}
        \centering
        \includegraphics[width=0.8\textwidth,trim={0 0cm 0cm 1cm},clip]
        {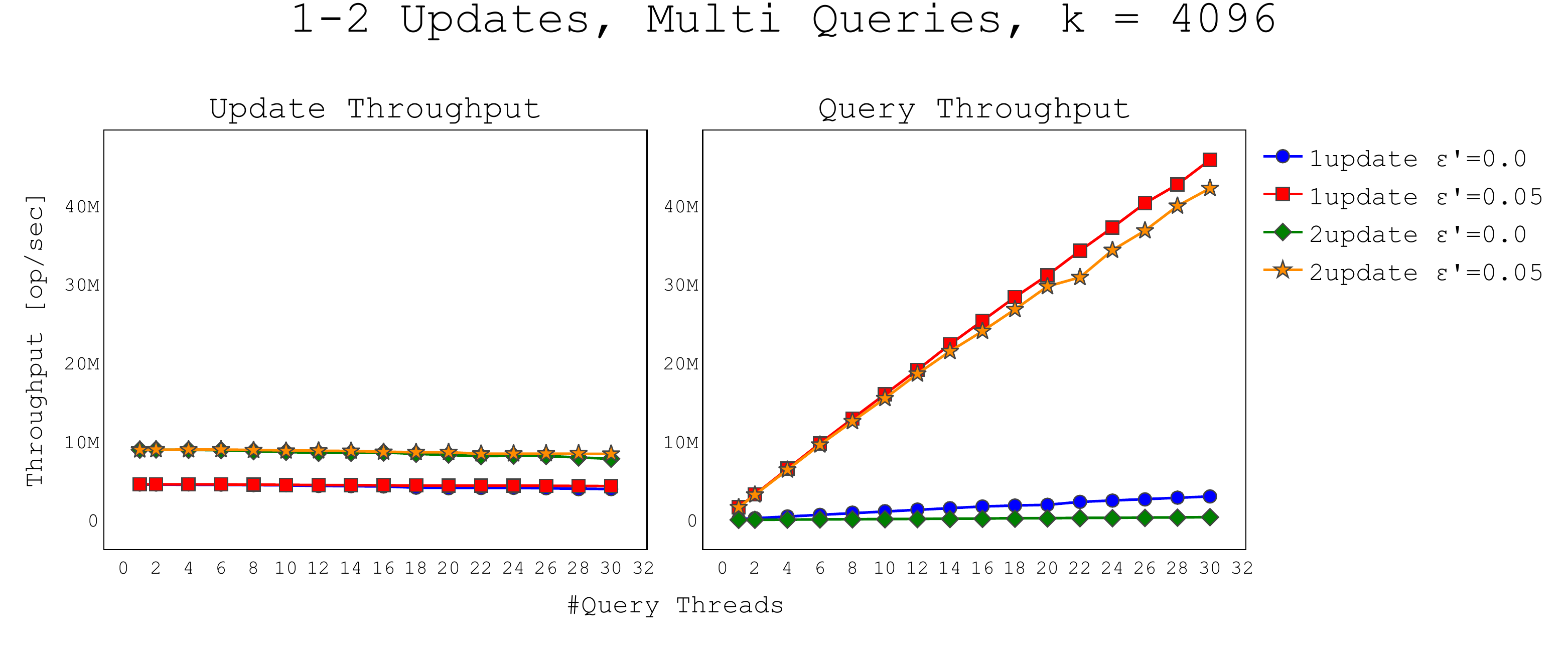}
        \caption{One or two update threads, up to 32 query threads, k=1024, 10M elements inserted after a pre-fill of 10M elements.}
        \label{fig:mixed_throughput}
    \end{subfigure}
    \caption{\mysketch throughput, k=4096, b=16.}
    \label{fig:throughput}
\end{figure*}

As shown in Figure \ref{fig:update_only_speedup}, \mysketch's performance with a single thread is similar to the sequential algorithm and with more threads it scales linearly, reaching $12x$ the sequential throughput with 32 threads. 
We observe that the speedup is faster with fewer threads, we believe this is because once there are more than $8$ threads, the shared object is accessed from multiple NUMA nodes.

Figure~\ref{fig:query_only_throughput} shows that, as expected, the throughput of the query-only workload scales linearly with the number of query threads, reaching $30x$ the sequential throughput with 32 threads.

In the mixed workload, the parameter $\rho$ is significant for performance - when $\rho = 0$ (no caching), a snapshot it reproduced on every query.
Figure~\ref{fig:mixed_throughput} presents the update throughput (left) and query throughput (right) in the presence of $1$ or $2$ update threads, with staleness thresholds of $\rho=0$ and $\rho=0.05$. We see that the caching mechanism ($\rho > 0$) is indeed crucial for performance. As expected, increasing the staleness threshold allows queries to use their local (possibly stale) snapshot, servicing queries faster and greatly increasing the query throughout. Furthermore, more update threads decrease the query throughput, as the update threads interfere with the query snapshot. 
Finally, increasing the number of query threads decreases the update throughput, as query threads interfere with update threads, presumably due to cache invalidations of the shared state.

\subsection{Parameter Exploration}
\label{ssec:params} 
We now experiment with different parameter settings with up to $32$ threads. 
In  Figure~\ref{fig: throughput_update_compare_k} we vary $k$ from $256$ to $4096$, in update-only scenario with $b = 16$ and up to $32$ update threads.
We see that the scalability trends are similar, and that \mysketch's throughput increases with $k$, peaking at $k = 2048$, after which increasing $k$ has little effect.
This illustrates the tradeoff between the sketch size (memory footprint) to  throughput and accuracy.

Figure~\ref{fig: throughput_update_compare_b} experiments with different local buffer sizes, from $1$ to $64$, in update-only scenario with $k = 4096$ and up to $32$ update threads. Not surprisingly, the throughput increases as the local buffer grows as this enables more concurrency. 

In Figure~\ref{fig: update_query_compare_rho} we vary $\rho$, in a mixed update-query workload with $8$ update threads, $24$ query threads, $k = 1024$, and $b = 16$, exploring another aspect of query freshness versus performance. As expected, increasing  $\rho$ has a positive impact on query throughput, as the cached snapshot can be queried more often.  Figure~\ref{fig: update_query_compare_rho} also shows the miss rate, which is the percentage of queries that need to re-construct the snapshot.

\begin{figure*}[]
    \begin{subfigure}{0.32\textwidth}
    \includegraphics[width=\textwidth,trim={0 0cm 1.9cm 2cm},clip]{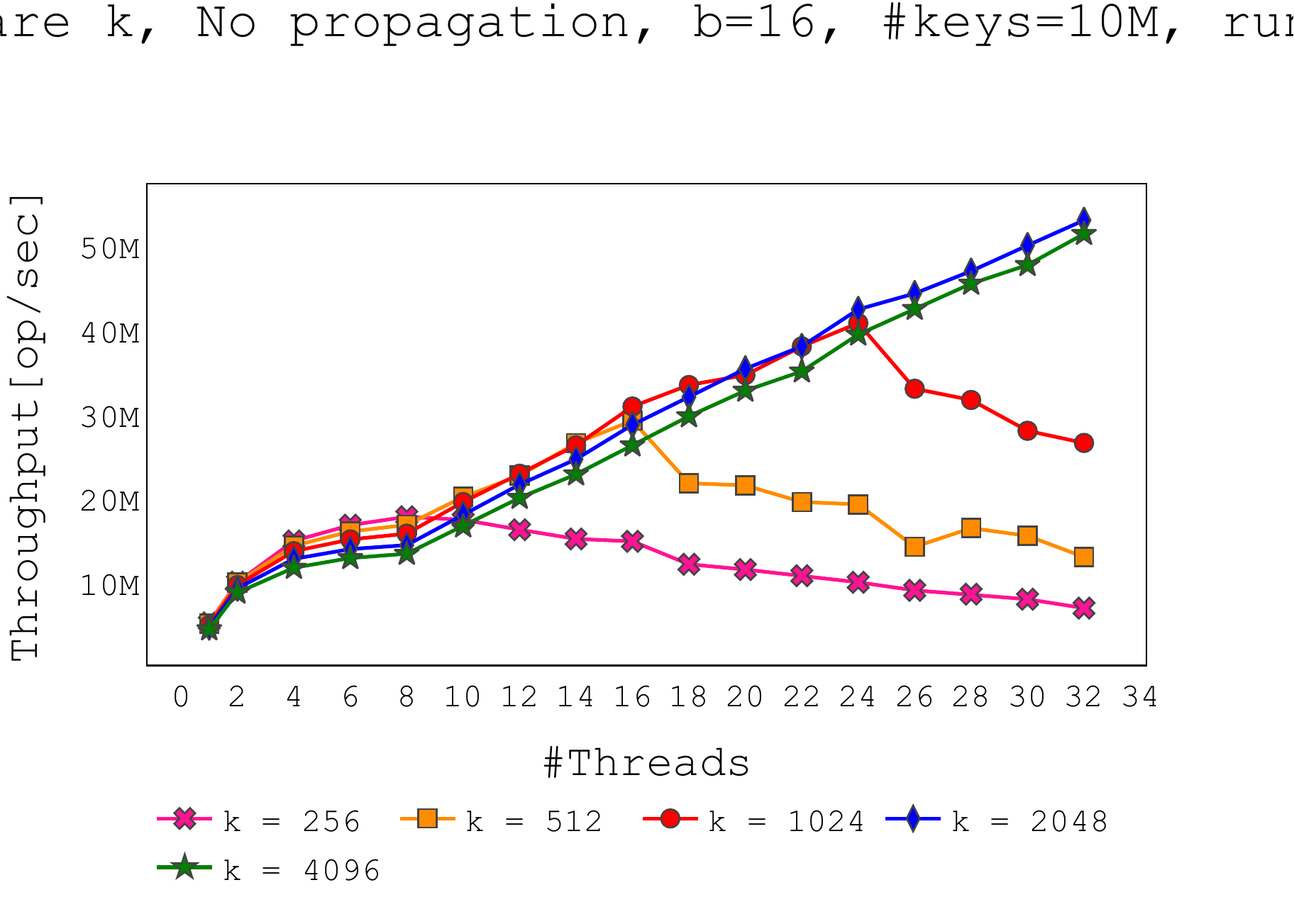}
    \caption{Update-only, \#keys=10M, b=16.}
    \label{fig: throughput_update_compare_k}
    \end{subfigure}
\hfill
    \begin{subfigure}{0.32\textwidth}
    \includegraphics[width=\textwidth,trim={0 0cm 1.9cm 2cm},clip]{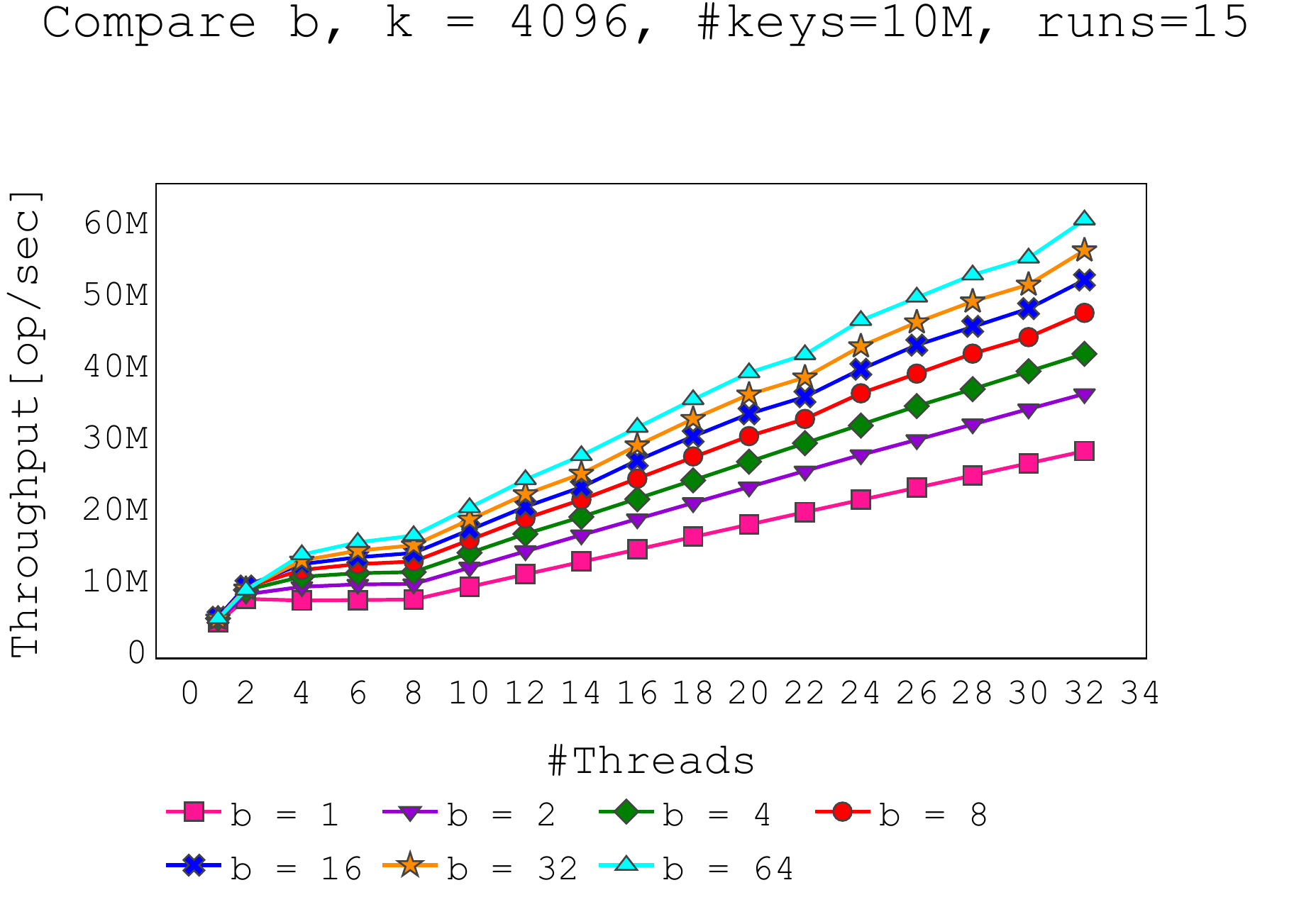}
    \caption{Update-only, \#keys=10M, k=4096.}
    \label{fig: throughput_update_compare_b}
    \end{subfigure}
\hfill
    \begin{subfigure}{0.32\textwidth}
    \includegraphics[width=\textwidth,trim={0 0cm 1.2cm 2cm},clip]
    {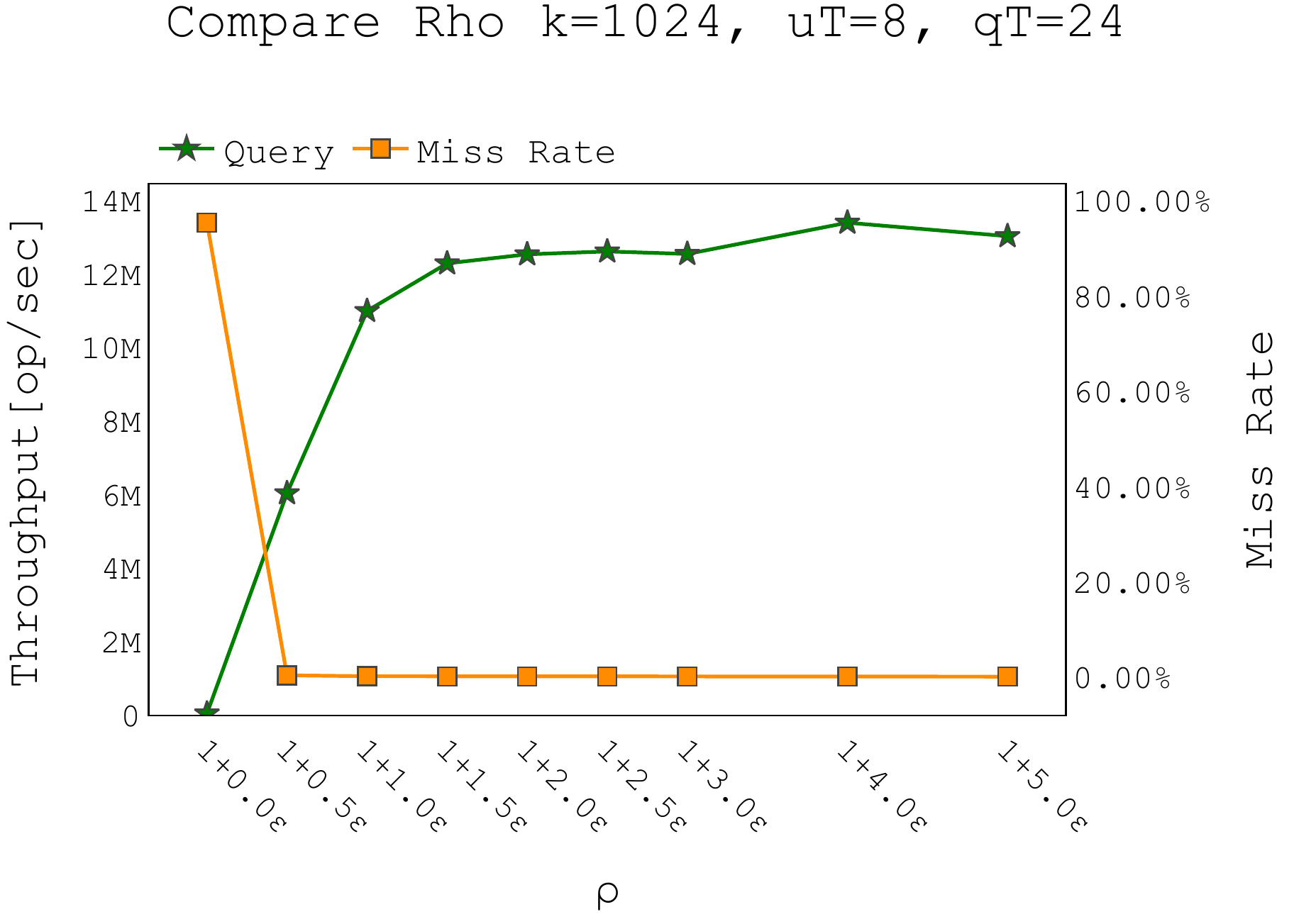}
    \caption{8 update threads, 24 query threads, \#keys=10M, k=1024 and b=16.}
    \label{fig: update_query_compare_rho}
    \end{subfigure}
\caption{\mysketch parameters impact.}
\label{fig: parameters_exploration}
\end{figure*}

\subsection{Accuracy}
\label{ssec:accuracy} 
To measure the estimate accuracy, we consider a query invoked in a quiescent state where no updates occur concurrently with the query. 
Figure~\ref{fig: accuracy_stderr} shows the standard error of 1M estimations in a quiescent state.
We see that \mysketch's estimations are similar to the sequential ones using the same $k$, and improves with larger values of $k$ as known from the literature on sequential sketches~\cite{mergeables_summaries}.
\begin{figure}[h]
    \centering
    \includegraphics[width=\columnwidth,trim={0 0cm 0cm 2cm},clip]
    {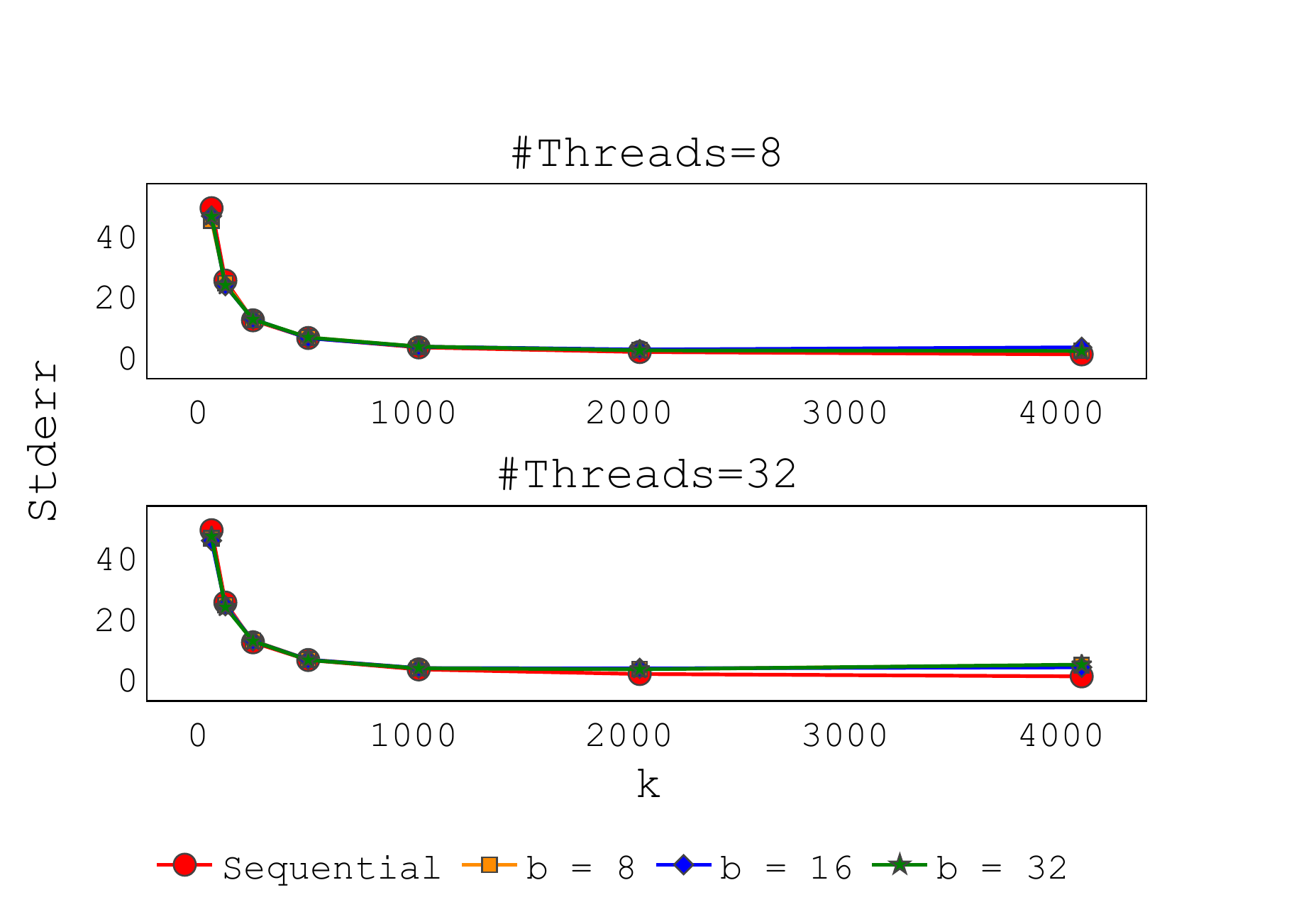}
    \caption{Standard error of estimation in quiescent state, keys=1M, runs=1000.}
    \label{fig: accuracy_stderr}
\end{figure}

To illustrate the impact of $k$ visually, Figure~\ref{fig:cdf} compares the distribution measured by \mysketch (red open-circles) to the exact (full information) stream distribution (green CDF filled-circles). In Figure~\ref{fig:intro-query-accuracy} (in the introduction), we depict the accuracy of \mysketch's estimate of a normal distribution with $k=1024$. Figure~\ref{fig:cdf_normal} (left) shows that when we reduce $k$ to $32$, the approximation is less tight while for $k=256$ (Figure~\ref{fig:cdf_normal} right) it is very accurate. We observe similar results for the uniform distribution in Figure~\ref{fig:cdf_uniform}. We experimented with additional distributions with similar results, which are omitted due to space limitations. 

\begin{figure}[h]
\centering
    \begin{subfigure}[]{\columnwidth}
        \centering
        \includegraphics[width=\columnwidth,trim={0.1cm 0.2cm 1.5cm 1cm},clip]
        {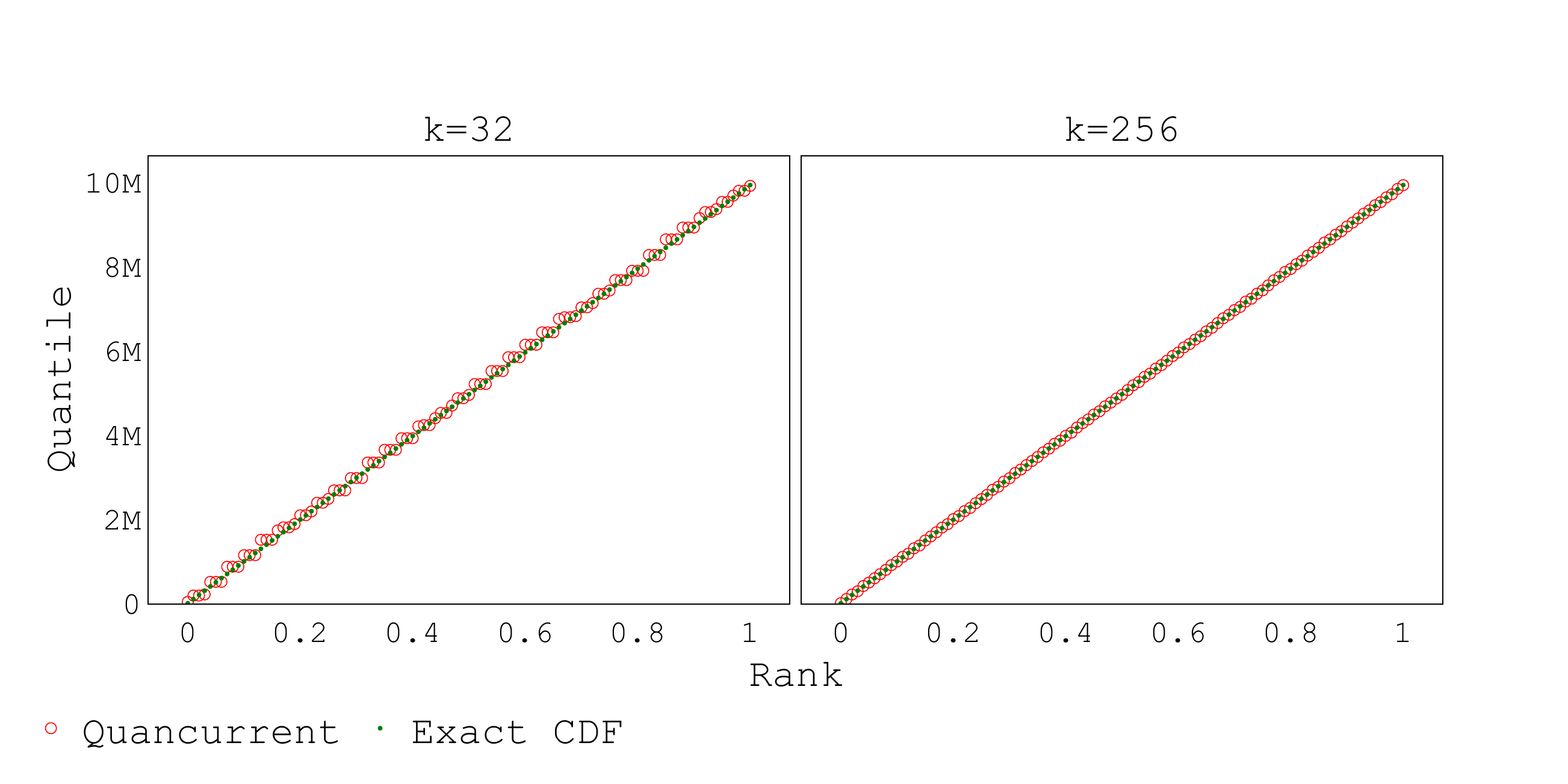}
        \caption{Uniform distribution.} \label{fig:cdf_uniform}
    \end{subfigure}
    
    \begin{subfigure}[]{\columnwidth}
        \centering
        \includegraphics[width=\columnwidth,trim={0.1cm 0.2cm 1.5cm 1cm},clip]
        {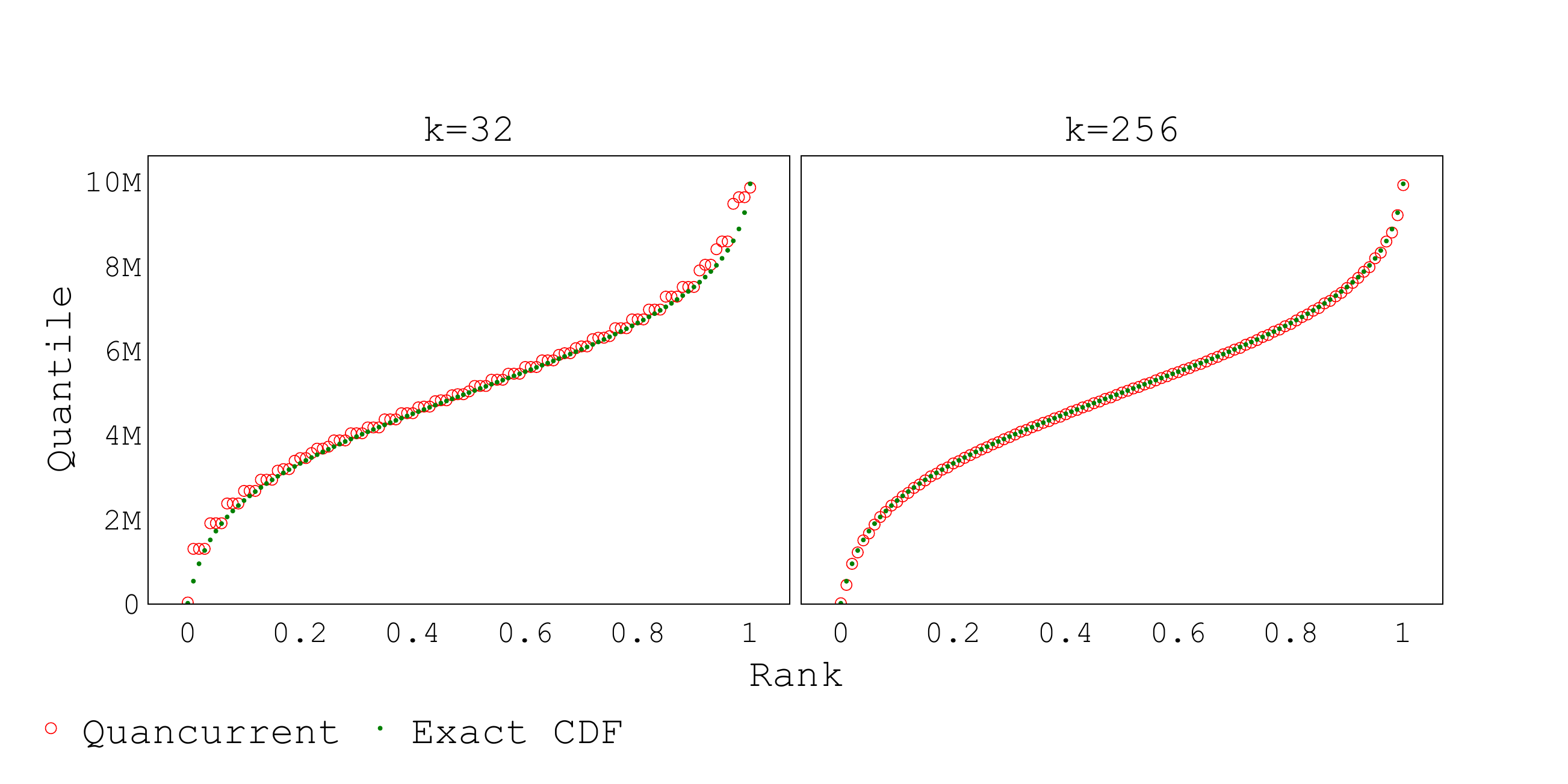}
        \caption{Normal distribution.} \label{fig:cdf_normal}
    \end{subfigure}

\caption{\mysketch quantiles vs. exact CDF, with 32 threads, b=16, and a stream size of 10M.} \label{fig:cdf}
\end{figure}

\subsection{Comparison to state of the art}
\label{ssec:compare} 
Finally, we compare \mysketch against a concurrent Quantiles sketch implemented within the FCDS framework~\cite{Rinberg_2020_fast_sketches}, the only previously suggested concurrent sketch we know that supports quantiles. Figure~\ref{fig:compare_FCDS_k4096} shows the throughput results (log scale) for $8$, $16$, $24$ and $32$ threads and $k=4096$.
FCDS satisfies relaxed consistency with a relaxation of up to $2NB$, where $N$ is the number of worker threads and $B$ is the buffer size of each worker. Recall that \mysketch's relaxation is at most $r = 4kS+(N-S)b$. For a fair comparison, we compare the two algorithms in settings with the same relaxation.

For $8$ update threads ($S=1$) and $b=2048$, the relaxation of \mysketch is $r\approx 30K$. The same relaxation in FCDS with the same number of update threads is achieved with a buffer size of $B=1920$.
With $8$ threads, \mysketch reaches a throughput of $22M\ ops/sec$ for a relaxation of $30K$ whereas FCDS reaches a throughput of $25M\ ops/sec$ for a much larger relaxation of $137K$. Also, with $32$ threads, \mysketch reaches a throughput of $62M\ ops/sec$ for a relaxation of $122K$, but FCDS only reaches a throughput of $19M\ ops/sec$ with a relaxation of more than $500K$.

Overall, we see that FCDS requires large buffers (resulting in a high relaxation and low query freshness) in order to scale with the number of threads. This is because, unlike \mysketch, FCDS uses a single thread to propagate data from all other threads' local buffers into the shared sketch. The propagation involves a heavy merge-sort, so large local buffers are required in order to offset it and keep the working threads busy during the propagation. In contrast, \mysketch's propagation is collaborative, with merge-sorts occurring concurrently both at the NUMA node level (in Gather\&Sort buffers) and at multiple levels of the shared sketch.

\begin{figure}[b]
\includegraphics[width=\columnwidth,trim={0cm 0.5cm 2cm 0.4cm},clip]
{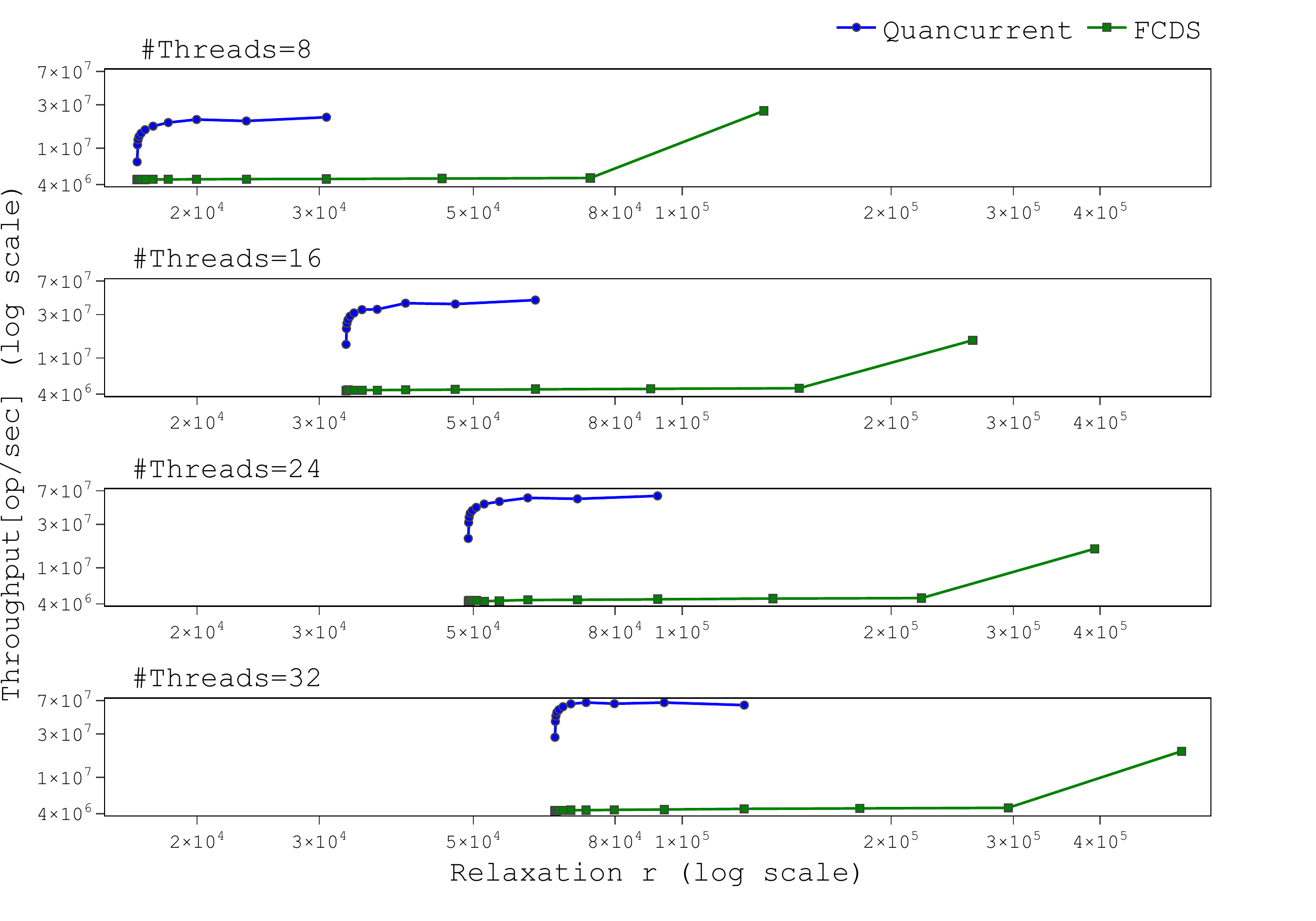}
\caption{\mysketch vs. FCDS, k = 4096.}
\label{fig:compare_FCDS_k4096}
\end{figure}






\section {Conclusion}
We presented Quancurrent, a concurrent scalable Quantiles sketch. We have evaluated it and shown it to be linearly scalable for both updates and queries while providing accurate estimates. Moreover, it achieves higher performance than state-of-the-art concurrent quantiles solutions with better query freshness.
Quancurrent’s scalability arises from allowing multiple threads to concurrently engage in merge-sorts, which are a sequential bottleneck in previous solutions.
We dramatically reduce the synchronization overhead by accommodating occasional data races that cause samples to be duplicated or dropped, a phenomenon we refer to as holes. This approach leverages the observation that sketches are approximate to begin with, and so the impact of such holes is marginal.
Future work may leverage this observation to achieve high scalability in other sketches or approximation algorithms.
\label{sec: conclusion}


\balance
\bibliographystyle{plain}
\bibliography{references}

\onecolumn
\newpage
\appendix

\section{Proofs}
\label{sup: sup_intro}
In Section~\ref{sup: prelim} we present preliminaries needed for our proofs. In Section~\ref{sup: query_proof} we prove that the query operation collects a consistent snapshot at some point in the execution. In Section~\ref{sup: correctness_proof} we prove that Quancurrent is strongly linearizable with respect to the relaxed specification. Finally, Section~\ref{sup: holes-analysis} provides proofs for the claims made in Section~\ref{sec:analysis}.

\subsection{Preliminaries}
\label{sup: prelim}
We consider a shared memory model, where a finite number of threads execute \emph{operations} on shared \emph{objects}. An operation consists of an \emph{invocation} and a matching \emph{response}. A \emph{history} \(\mathnormal{H}\) is a finite sequence of operation invocation and response steps. A history \(\mathnormal{H}\) defines a partial order \(\prec_\mathnormal{H}\) on operations: Given operations \(op\) and \(op'\), \(op \prec_H op'\) if and only if \(response(op)\) precedes \(invocation(op')\) in \(\mathnormal{H}\). Two operations that do not precede each other are \emph{concurrent}. In a \emph{sequential history}, there are no concurrent operations. An object is specified using a \emph{sequential specification} $\mathcal{H}$, which is the set of its allowed sequential histories.
An operation \emph{op} is \emph{complete} in a history \(\mathnormal{H}\) if both \(invocation(op)\) and its matching \(response(op)\) are in \(\mathnormal{H}\).
A \emph{linearization} of a concurrent history \(\mathnormal{H}\), is a sequential history $H'$ such that: (1) $H' \in \mathcal{H}$, (2) $H'$ contains all completed operations and possibly additional non-complete ones, after adding matching responses, and, (3) $\prec_{H'}$ extends $\prec_H$. A correctness condition for randomized algorithms \emph{strong linearizability}~\cite{strong_linearizability}, defined as follows: 
\begin{definition}[Strong linearizability] \label{Def: strong_linearizability}
A function \(f\) mapping executions to histories is prefix-preserving if for any two executions \s, \s', where \s is a prefix of \s', \fs is a prefix of \fstag.

An object A is a strongly linearizable if there is a prefix-preserving function \(f\) that maps every history \(\mathnormal{H}\) of $A$ to a linearization of \(\mathnormal{H}\).
\end{definition}

Our algorithm is randomised and we consider a weak adversary that determines the scheduling without observing the coin-flips.

As previously mentioned, we adopt a flavor of \emph{relaxed semantics}, as defined in \cite{Henzinger_2013_Quantitative_Relaxation}:
\begin{definition}[$r$-relaxation] \label{Def: relaxation}
A sequential history $H$ is an
$r$-relaxation of a sequential history H', if H is comprised of
all but at most $r$ of the invocations in $H'$ and their responses,
and each invocation in $H$ is preceded by all but at most $r$ of the invocations that precede the same invocation in $H'$.
The $r$-relaxation of a sequential specification $\mathcal{H}$ is the set of histories that have $r$-relaxations in $\mathcal{H}$:
\[{H}^r \triangleq \left\{ H' \mid \exists H \in \mathcal{H}: H \text{ is an } r\text{-relaxation of } H'\right\}. \]
\label{def:r-relaxtion}
\end{definition}


\newpage
\subsection{Query Proof}
\label{sup: query_proof}
Let $\sigma$ be an execution of \mysketch and let $tm1$ and $tm2$ be responses of two read operations of $tritmap$, $op_1$ and $op_2$, respectively, in $\sigma$, such that $op_1$ precedes $op_2$. Let $\mathnormal{A}$ be the stream represented by $tm1$. We will prove that the constructed \emph{snapshot} in Algorithm~\ref{alg: sl_query}, Lines~\ref{line: query_estimate_start}-\ref{line: query_estimate_end}, summarizes the same stream $\mathnormal{A}$. We denote by $|\mathnormal{A}|$ to be the stream size and by $|snapshot|$ to be the size of the stream currently represented by $snapshot$.
Let \emph{snapLevels} be the set all levels read between $op_1$ and $op_2$.

First, we show that $tm1$ and $tm2$ represent the same stream, i.e.  $\mathnormal{A}$.

\begin{lemma}
Let $tm1$ and $tm2$ be responses of two read operations of $tritmap$.
If $tm1$ and $tm2$ represent streams with equal size then $tm1$ and $tm2$ represent the same stream.
\end{lemma}
\begin{proof}
By definition the variable tritmap is 3-based, 31 digits non-negative integer that is atomically updated by DCAS operation and is only increasing.   
\end{proof}

Second, we show that $snapLevels$ contains all sampled elements in the sketch, immediately after $tm2$, summarizing the stream $\mathnormal{A}$.

\begin{lemma}\label{Lem: snap_no_miss}
The set of levels, $snapLevels$, read between the two \emph{tritmap}'s reads, $op_1$ and $op_2$, contains all sampled elements summarizing the stream $\mathnormal{A}$.
\end{lemma}
\begin{proof}
Let $x$ be an element in level $j$ immediately after $tm1$.
If $x$ exists in level $j$ during the read of sketch's levels in Algorithm~\ref{alg: sl_query} Line~\ref{Line: read_snap}, then we are done.
Otherwise, a propagation occurred in between the reads such that level $j$ was merged with the next level and cleared. During this merge, all level $j$'s elements, including $x$, were sampled and propagated to level $j+1$. If $x$ exists in level $j+1$ in the set $snapLevels$ then we are done, if not, we apply the above argument again. This continues up to MAX\_LEVEL and therefore snap contains the element $x$.
\end{proof}

The following refers to the construction of the subset \emph{snapshot} from level MAX\_LEVEL to 0:

\begin{lemma} \label{Lem: all_dup}
If level $j$ contains an element represented by the snapshot thus far, all elements in level $j$ are also (already) represented by this snapshot.
\end{lemma}
\begin{proof}
During a call to the procedure \emph{propagate} with level $j$, the elements in that level are sampled and merged with level $j+1$, resulting in level $j+1$ representing also the elements of (former) level $j$.
Let $x$ be an element represented by level $j$ in the set $snapLevels$.
$x$ is also represented by level $i$ such that level $i$ is in the current \emph{snapshot} and $i>j$. Consider the process of propagation from level $j$ to level $i$, it follows from the above that level $i$ also represents all the elements in level $j$.
\end{proof}

Each element in level $j$ represents $2^i$ elements from the processed stream.

\begin{lemma}\label{Lem: all_in_no_dup}
If level $j$ represents new elements that are not represented by snapshot, the size of the sub-stream left to represent is at least the size of the representation of level $j$.

If level $j$ represents elements already represented by snapshot, the sub-stream left to represent is smaller than the representation of level $j$.
\end{lemma}
\begin{proof}
We will prove the first part by contradiction. Assume by contradiction that the size of sub-stream left to represent is smaller than the size of the representation of level $j$ such that $|snapshot\cup level[j]| > |\mathnormal{A}|$. It follows from Lemma~\ref{Lem: snap_no_miss} that level $j$ contains at least one duplicated element (an element already represented by level $i$ in \emph{snapshot} such that $i>j$). From Lemma~\ref{Lem: all_dup} all level $j$'s elements are duplicated and already being represented by the current \emph{snapshot}. Contradiction. 

Now, for the second part, assume that an element $x$ is represented more than once by the current \emph{snapshot}. Let level $j$ be the second level to be added to \emph{snapshot} that also represents $x$ (first time to duplicate the representation of $x$). From Lemma~\ref{Lem: all_dup}, all elements in level $j$ are already represented by this snapshot.
If level $j$ contains $2k$ elements, the size of sub-stream represented by level $j$ is $2k\cdot2^j$, and the size of the sub-stream left to represent is at most represented by levels 0 to $j-1$, meaning the size is at most $2k(1+2+\dots+2^{j-1})=2k(2^j-1)$ which is smaller than the size of sub-stream represented by level $j$.
If level $j$ contains $k$ elements and is already represented by level $i$, $i>j$, in \emph{snapshot}, level $j-1$ must have propagated level $j$ deeper, i.e. level $j-1$ is also represented by level $i$, therefore, also already represented by \emph{snapshot}. Level $j$ represents $k\cdot2^j$. The size of the sub-stream left to represent is at most represented by levels 0 to $j-2$, meaning the size is at most $2k(1+2+\dots+2^{j-2})=k(2^{j-1}-2)$, which is smaller than the size of sub-stream represented by level $j$.
\end{proof}

\begin{lemma} \label{Lem: query_estimate}
The snapshot constructed in Algorithm~\ref{alg: sl_query}, Lines~\ref{line: query_estimate_start}-\ref{line: query_estimate_end}, summarizes the same stream, $\mathnormal{A}$, as represented by the second tritmap read, $tm2$, in Algorithm~\ref{alg: sl_query}, Line~\ref{Line: second-collect}. 
\end{lemma}
\begin{proof}
From Lemma~\ref{Lem: all_in_no_dup} $|snapshot|=|\mathnormal{A}|$ and every element is represented at most once in the constructed \emph{snapshot}.
\end{proof}


\newpage
\subsection{Correctness Proof}
\label{sup: correctness_proof}
Queries are answered from an array of ordered tuples summarizing the total stream processed so far and denoted as $\mathit{samples}$. Each tuple contains a summary point (i.e.,a value from the sketch) and its associated weight. The samples array contains all the sketch's summary points and is sorted according to their values. Note that, only levels with $\mathit{tritmap}[i]\in{1,2}$ are included.

As described in Section~\ref{Section: concurrent_algorithm}, the update operation is divided into 3 stages: 
(1) $\mathit{gather\ and\ sort}$ is the process of ingesting stream elements into a $\mathit{Gather\&Sort}$ unit.
(2) $\mathit{batch\ update}$ is the process of copying $2k$ elements from one of the \emph{G\&SBuffer}s into \mysketch's first level. 
(3) $\mathit{propagate\ levels}$ is the process of merging base level up the sketch's levels until reaching an empty level. 


Correctness of an object's implementation is defined with respect to a sequential specification $\mathcal{H}$. Sequential specification is defined with respect to deterministic objects. Therefore, we de-randomized the Quantiles sketch by providing coin flips with every update. We call the set of sequential histories of the deterministic Quantiles sketch as \emph{SeqSketch}.

Rinberg et al.~\cite{Rinberg_2020_fast_sketches} defined the relation between a sequential history and a stream:
\begin{definition} 
Given a finite sequential history H, $\mathcal{S}(H)$ is the stream $a_1,\dots,a_n$ such that $a_k$ is the argument of the $k^{th}$ update in H.
\end{definition}

The notion of \emph{happens before} in a sequential history as defined in \cite{Rinberg_2020_fast_sketches}:
\begin{definition}
Given a finite sequential history $\mathit{H}$ and two method invocation $M_1,M_2$ in $\mathit{H}$, if $M_1$ precedes $M_2$ in $\mathit{H}$, we denote $M_1 \prec_H M_2$.
\end{definition}

\begin{definition}[Unprop updates] \label{Def: unprop_update}
Given a finite execution \s of \mysketch, we denote by suffix(\s) as the suffix of \s starting at the last successful $\mathit{batch update}$ event, or the beginning of \s if no such event exists.
We denote by $\mathit{up\_suffix}(\sigma)$ the sub-sequence of $\mathit{H(suffix}(\sigma))$ consisting of updates operations in the $Gather\&Sort$ units.
We denote by $\mathit{up\_suffix_i}(\sigma)$ the sub-sequence of $\mathit{H(suffix}(\sigma))$ consisting of updates operations in the local buffer of thread $T_i$.
\end{definition}

\begin{definition}[Updates Number] \label{Def: updates_num}
We denote the number of updates in history $\mathit{H}$ as $\mathit{|H|}$.
\end{definition}


\begin{lemma} \label{Lem: sl_relaxation}
\mysketch is strongly linearizable with respect to the relaxed specification $SeqSketch^r$ with \(r=4kS + (N-S)b\), where $S$ is the number of NUMA nodes, $k$ is the sketch summary size, $b$ is the size of threads local buffer and $N$ is the number of update threads.
\end{lemma}
\begin{proof}
\mysketch is an $r$-relaxed concurrent Quantiles sketch. The correctness condition for randomized algorithms under concurrency is strong linearizability~\cite{strong_linearizability}. Strong linearizability is defined with respect to the sequential specification of a data structure. We denote by $\mathit{SeqSpec}$ the sequential specification of \mysketch.

A relaxed consistency extend the sequential specification of an object to a larger set that contains sequential histories which are not legal but are at bounded "distance" from a legal sequential history~\cite{Henzinger_2013_Quantitative_Relaxation,Afek_2010_Quasi_linearizability,Rinberg_2020_fast_sketches}. We convert \mysketch into a deterministic object by providing a coin flip with every update. We re-define (de-randomized) \mysketch sequential specification by relaxing it. Intuitively, we allow a query to "miss" a bounded number of updates that precede it. Quantiles sketch is order agnostic, thus re-ordering updates is also allowed. 

Let $\sigma$ be a concurrent execution of \mysketch. We use two mappings from concurrent executions to sequential histories defined as follows.
We define a mapping, $l$, from a concurrent execution to a serialization, by ordering operations according to the following linearization points:
\begin{itemize}
\item \textbf{Query} linearization point is the second \emph{tritmap} read, $tm2$, such that it summarizes the same stream size as $tm1$ (Algorithm~\ref{alg: sl_query}, Line~\ref{Line:query_linearization}).
\item \textbf{Update} linearization point is the insertion to threads local buffers (Algorithm \ref{alg: gather}, Line \ref{Line: update_linearization}).
\end{itemize}

Strong linearizability requires that the linearization of a prefix of a concurrent execution is a prefix of the linearization of the whole execution. By definition, $l(\sigma)$ is prefix-preserving. Note that $l(\sigma)$ is a serialization that does not necessarily meets the sequential specification.

Relaxed consistency extends the sequential specification of an object to include also relaxed histories.
We define a mapping, $f$, from a concurrent execution to a serialization, by ordering operations according to visibility points:
\begin{itemize}
\item \textbf{Query} visibility point is its query's linearization point.
\item \textbf{Update} visibility point is the time after its invocation in $\sigma$ such that the \emph{G\&SBuffer} (this update is inserted into) is batched updated into level 0 with DCAS. If there is not such time, then this update does not have a visibility point, meaning, it is not included in the relaxed history, $f(\sigma)$.
\end{itemize}

To prove correctness we need to show that for every execution $\sigma$ of \mysketch: (1) $f(\sigma) \in \mathit{SeqSpec}$, and (2) $f(\sigma)$ is an $r$-relaxation of $l(\sigma)$ for $r=4kS + (N-S)b$.

We show the first part. 
\begin{lemma} \label{Lem: visibility_in_seq_spec}
Given a finite execution $\sigma$ of \mysketch, $f(\sigma)$ is in the sequential specification. 
\end{lemma}

\begin{proof}
First, we present and prove some invariants. 

\begin{invariant}\label{Inv: GSBuffer_summary}
The Gather\&Sort object summarises at most 4k elements.
\end{invariant}
\begin{proof}
The Gather\&Sort unit contains two buffers of \(2k\) elements. Elements are ingested into the buffer without a sampling process. The desired summary is agnostic to the processing order, therefore \(\mathnormal{S}\) summarises history of \(4k\) update operations and their responses. 
\end{proof}

\begin{lemma}
The variable tritmap is a monotonic increasing integer.
\end{lemma}
\begin{proof}
The variable tritmap is altered only in Line~\ref{Line:insert_batch} of Algorithm~\ref{alg: batch_update}, in Line~\ref{Line:next_full_DCAS} of Algorithm~\ref{alg: propagate} and in Line~\ref{Line:next_empty_DCAS} of Algorithm~\ref{alg: propagate}. By definition, it is only incremented.
\end{proof}


\begin{invariant} \label{Inv: tritmap_sketch_state}
The variable tritmap represents the sketch state:
\begin{itemize}
    \item If \(tritmap[i] = 0\), then \(levels[i]\) is empty or does not contained in the sketch's samples array.
    \item If \(tritmap[i] = 1\), then \(levels[i]\) contains \(k\) points associated with a weight of \(2^i\).
    \item If \(tritmap[i] = 2\), then \(levels[i]\) contains \(2k\) points associated with a weight of \(2^i\).
\end{itemize}
\end{invariant}
\begin{proof}
The proof is by induction on the length of levels array (or the current maximum depth of levels[]).\\
\underline{Base}: By definition, tritmap is initialized to 0 and updated only at the batchUpdate procedure and the propagate procedure.
After the first batch update, level 0 contains $2k$ elements and tritmap is increased by 2 such that tritmap[0]=2. When this first batch is merged with the next level, level 1 contains $k$ elements and tritmap is increased by 1 such that tritmap[0]=0.
On each propagation, we first perform a batch update of one of the G\&SBuffer arrays to level 0 and increase tritmap by 2. Then we call propagate() starting with level 0. Level 0 is merged with the next level and tritmap is incremented by 1. Therefore, after each batchUpdate, \(tritmap[0] = 2\) and level 0 contains \(2k\) elements and after each call to propagate(0) \(tritmap[0] = 0\) and level 0] is not contained in the sketch's samples array. The following calls to propagate increase tritmap by $3^i$ for $i>0$ and \(tritmap[0] = 0\) until the end of the current propagation. \\
\underline{Inductive hypothesis}: We assume the invariant holds for all levels i such that \(i>0\) and prove it holds for level \(i+1\). By definition, tritmap is updated only at batchUpdate and propagate procedures. For \(i>0\), \(tritmap\) is changed only if \(tritmap[i]=2\). By the inductive hypothesis, if \(tritmap[i] = 2\), then \(levels[i]\) contains \(2k\) points associated with a weight of \(2^i\). If propagation has not yet reached level i+1, it is empty and \(tritmap[i+1]=0\) from initialization. After a call to propagate(i), \(levels[i+1]\) contains k points associated with a weight of \(2^{i+1}\) and tritmap satisfies \([b_{31},\dots,b_{i+2},0,2,b_{i-1},\dots,b_0] + 3^i = [b_{31},\dots,b_{i+2},1,0,b_{i-1},\dots,b_0]\) i.e \(tritmap[i+1]=1\). Next time propagation will reach level i+1, it will contain \(2k\) points associated with a weight of \(2^{i+1}\) and tritmap will satisfy \([b_{31},\dots,b_{i+2},1,2,b_{i-1},\dots,b_0] + 3^i = [b_{31},\dots,b_{i+2},2,0,b_{i-1},\dots,b_0]\) i.e \(tritmap[i+1]=2\). Note that each propagation starts from level 0 and stops when reaching an empty level $j$, the tritmap trit larger then $j$ are not changed.
\end{proof}

\begin{invariant}\label{Inv: summary_history}
Given a finite execution $\sigma$ of \mysketch, it summarises \(f(\sigma)\).
\end{invariant}
\begin{proof}
The proof is by induction on the length of \(\sigma\).\\
\underline{Base}: The base is immediate. \(\mathnormal{S}\) summarises the empty history.\\
\underline{Inductive hypothesis}: We assume the invariant holds for \(\sigma'\), and prove it holds for \(\sigma = \{\sigma',step\}\). We consider only steps that can alter the invariant, meaning steps that can change the sketch state.
\begin{itemize}
    \item DCAS operation in batchUpdate, increasing tritmap by 2 and copying one of the G\&SBuffer arrays into the first level of \mysketch.
    \item[] By the inductive hypothesis, before the step, \mysketch summarises \(f(\sigma')\). If the DCAS fails, the sketch state has not change. Else, $2k$ elements were copied to the level 0 and tritmap was increased by 2. 
    From Invariant \ref{Inv: GSBuffer_summary}, a G\&SBuffer array summarises a collection of \(2k\) updated elements \(\{a_1,\dots,a_{2k}\}\). By copying, we sequentially ingest the stream \(B=\{a_1,\dots,a_{2k}\}\) to \mysketch. Let \(A=\mathcal{S}(f(\sigma'))\). By definition, \mysketch summarises \(A||B\). Therefore \mysketch summarises \(f(\sigma)\), preserving the invariant. 
    \item DCAS operations in propagate, updating tritmap and merging level \(i\) with its following level.
    \item[] By the inductive hypothesis, before the step, \mysketch summarises \(f(\sigma')\). If the DCAS fails, the sketch state has not changed by the step. Else, we propagated level \(i\) into level \(i+1\). By definition, \(k\) points from level \(i\) were merged with level \(i+1\), with the weight of each point scaled up by a factor 2. tritmap[i]=0 and therefore level \(i\) was disabled from \(samples[]\) such that \(2k\) points associated with \(2^i\) weight are not included in the summary. The total weight of the summary points was not changed. The sketch's state summarises the same stream, no new points were added and the stream size was not changed. \mysketch summarises \(f(\sigma)\), preserving the invariant. 
    \item Operations to clear level i, updating \(levels[i] \gets \bot\).
    \item[] By definition, \(tritmap[i] = 0\). By Invariant \ref{Inv: tritmap_sketch_state}, \(levels[i]\) is empty or not included in \(samples[]\), meaning clearLevel does not affect the summary points.  \(\mathnormal{S}\) summarises \(f(\sigma)\), preserving the invariant. 
\end{itemize}
\end{proof}

\begin{lemma}[Query Correctness] \label{Lem: query_correctness}
Given a finite execution \s of \mysketch, let Q be a query that returns in \s. Let \v be the visibility point of \Q, and let $\sigma'$ be the prefix of \s until point \v. Q returns a value equal to the value returned 
by a sequential sketch after processing $\mathcal{S}\left(f(\sigma')\right)$.

\end{lemma}
\begin{proof}
Let \s be an execution of \mysketch, let \Q be a query that returns in \s, and let \v be the visibility point of \Q. Let \s' be the prefix of \s until point \v, and let \(A=\mathcal{S}(f(\sigma'))\).
By definition, the visibility point of a query is when the second tritmap read returns a value representing the same stream size as the previous tritmap read. As proved in Lemma~\ref{Lem: query_estimate}, the collected state represents the same stream as \mysketch at the visibility point. By Invariant \ref{Inv: summary_history}, at point \v, \mysketch summarises \fstag, and, similarly, summarises the stream \A. 
Therefore, \(query(arg)\) returns a value equal to the value returned by a sequential sketch after processing \(A=\mathcal{S}(f(\sigma'))\).
\end{proof}

We have shown that each query in \fs estimates all updates that happened before its invocation. Specifically, a query invocation at the end of a finite execution \s,
returns a value equal to the value returned by a sequential sketch after processing
\A. By this, we have proven that $ f(\sigma) \in SeqSpec$.
\end{proof}

We now show that for every execution \s, \fs  is an \(r\)-relaxation of \(l(\sigma)\) for \(r = 4kS + (N-S)b\).
The order between operations satisfies:

\begin{lemma}\label{Lem: op_order}
Given a finite execution \s of \mysketch, and given an operation O (query or update) in l(\s), for every query Q in l(\s) such that Q happened before O in \(l(\sigma)\), then Q happened before O in \(f(\sigma)\):  \[Q \prec_{l(\sigma)} O \Rightarrow  Q \prec_{f(\sigma)} O\]
\end{lemma}
\begin{proof}
If O is a query then the proof is immediate since the visibility point and the linearization point of query are equal. Else, O is an update. By definitions, the linearization point of update happens before its visibility point. As the linearization point and visibility point of query \Q are equal, it follows that if \(Q \prec_{l(\sigma)} O\) then \(Q \prec_{f(\sigma)} O\).
\end{proof}

Note that as query linearisation point is equal to its visibility point, all queries in \fs will also be in \ls. 

\begin{lemma}\label{Lem: GSBuffer_updates_num}
Given a finite execution \s of \mysketch, the maximum number of unpropagated updates operations in Gather\&Sort units is \(S \cdot 4k\): \[|up\_suffix(\sigma)| \leq S \cdot 4k\], where S is the number of NUMA nodes
\end{lemma}
\begin{proof}
If update operation is included in \(up\_suffix(\sigma)\), the size of the array in G\&SBuffer that the update is a member of, is less-equal $2k$. By definition, if both arrays in a Gather\&Sort unit are full, no update thread (pinned to the same node as the Gather\&Sort unit) can copy his local buffer's elements. It follows that \(|up\_suffix(\sigma)| \leq S\cdot4k\).
\end{proof}

We give an upper bound on the number of updates in a threads local buffers.

\begin{lemma}\label{Lem: local_updates_num}
Given a finite execution \s of \mysketch, the number of unpropagated updates in the local buffer of thread \(T_i\) is bounded by b, \[|up\_suffix_i(\sigma)| \leq b\]
\end{lemma}
\begin{proof}
If update is included in \(up\_suffix_i(\sigma)\) , it follows that \(|items_buf_i| \leq b\) and therefore \(|up\_suffix_i(\sigma)| \leq |items_buf_i| \leq b \). When the local buffer of thread \(T_i\) is full, it copies \(items_buf_i\) to one of the G\&SBuffer's arrays and the corresponding updates will not be included in \(up\_suffix_i(\sigma)\).
\end{proof}

To prove that \fs is an \((4kS + (N-S)b)\)-relaxation of \ls, first, we will show that \fs comprised of all but at most \(r=4kS + (N-S)b\) invocations in \ls and their responses.

\begin{lemma} \label{Lem: invocations_bound}
Given a finite execution \s of \mysketch, \[ |f(\sigma)| \ge |l(\sigma)| - (4kS + (N-S)b) \]
\end{lemma}
\begin{proof}
\ls contains all updates, \fs contains all updates with visibility points. Updates without visibility points are the unpropagated updates in G\&SBuffers and unpropagated updates in the local buffer of each update thread. There are \(N\) update threads, therefore, excluding S update thread that may continue, \(|f(\sigma)| = |l(\sigma)| - (\sum_{i=1}^{N-S}|up\_suffix_i(\sigma)|) - |up\_suffix(\sigma)|\). From Lemma \ref{Lem: local_updates_num}, \(|up\_suffix_i(\sigma)| \leq b\) and from Lemma \ref{Lem: GSBuffer_updates_num}, \(|up\_suffix(\sigma)| \leq S \cdot 4k\). Therefore, \(|f(\sigma)| \ge |l(\sigma)| - (4kS + (N-S)b)\).
\end{proof}

To complete the poof that \fs is an \((4kS + (N-S)b)\)-relaxation of \ls, we will show that each invocation in \fs is preceded by all but at most \((4kS + (N-S)b)\) of the invocations that precede the same invocation in \ls.

\begin{lemma}\label{Lem: relaxation}
Given a finite execution \s of \mysketch, \fs is an \((4kS + (N-S)b)\)-relaxation of \ls.
\end{lemma}
\begin{proof}
Let O be an operation in \fs such that O is also in \ls. Let \(Ops\) be a collection of operations preceded O in \ls but not preceded O in \fs, i.e \( Ops=\{O'| O' \prec_{l(\sigma)} O \wedge  O' \nprec_{f(\sigma)} O \} \). By Lemma \ref{Lem: op_order}, query \( Q \notin Ops \). Let \(\sigma^{pre}\) be the prefix of \s and let \(\sigma^{post}\) be the suffix of \s such that \( l(\sigma)=\sigma^{pre},O,\sigma^{post} \). From Lemma \ref{Lem: local_updates_num}, \(|f(\sigma^{pre})| \ge |l(\sigma^{post})| - (4kS + (N-S)b))\). As \(|f(\sigma^{pre})|\) is the number of updates preceded O in \(f(\sigma^{pre})\), and \(|l(\sigma^{pre})|\) is the number of updates preceded O \(l\sigma^{pre})\), it follows that \(|Ops| = |l(\sigma^{pre})|-|f(\sigma^{pre})| \leq |l(\sigma^{pre})|-(|l(\sigma^{post})| - (4kS + (N-S)b))\leq (4kS + (N-S)b)\). Therefore, by Definition \ref{Def: relaxation}, \fs is an \((4k+b(N-1))\)-relaxation of \ls.
\end{proof}

Finally, we have proven that given a finite execution \s of \emph{Quancurrent}, \ls is strongly linearizable, \(f(\sigma) \in SeqSpec\) and \fs is an \((4kS + (N-S)b)\)-relaxation of \ls. We have proven Lemma \ref{Lem: sl_relaxation}.

\end{proof}

\newpage
\subsection{Hole Analysis Proofs}
\label{sup: holes-analysis}
To show the bound on the expected number of holes, we first show the following claim.
\begin{claim}
$\left(\frac{1}{2}\right)^{jb + 2i +1} {{jb+2i} \choose i}$ is monotonically increasing for $i \in \{0,1,\dots,b-1\}$ for $j \in \mathds{N}.$
\label{clm:supp-helper-claim-monotonic}
\end{claim}
\begin{proof}
Denote by $f(i)$:
\[\left(\frac{1}{2}\right)^{jb + 2i +1} {{jb+2i} \choose i}.\]

For $b=1$ the claim is immediate.
For $b=2$:
\[f(1)=\left(\frac{1}{2}\right)^{2j+3} {{2j+2} \choose 1} = f(0) \cdot \frac{2j+2}{4} \geq f(0) \cdot \frac{2+2}{4} = f(0).\]

For $b>2$, we show that $f(i+1)\geq f(i)$ for all $0 \leq i \leq b-2$.

\begin{align*}
    f(i+1) &= \left(\frac{1}{2}\right)^{jb + 2i +3} {{jb+2i+2} \choose {i+1}} \\
    &= \left(\frac{1}{2}\right)^{jb + 2i +3} \cdot \frac{(jb+2i+2)!}{(jb+i+1)!\cdot(i+1)!} \\
    &= f(i) \cdot \frac{1}{4} \cdot \frac{(jb+2i+2)(jb+2i+1)}{(jb+i+1)(i+1)}
\end{align*}
Consider the following function:
\[g(x) = \frac{1}{4} \cdot \frac{(jb+2x+2)(jb+2x+1)}{(jb+x+1)(x+1)}.\]
It is monotonically decreasing for $0 \leq x \leq b-1$. Consider $g(b-2)$:
\[\frac{1}{4} \cdot \frac{(jb+2b-2)(jb+2b-3)}{(jb+b-1)(b-1)}.\]
Denote this function as $h(j)$.
\[h(1) = \frac{1}{4} \cdot \frac{(3b-2)(3b-3)}{(2b-1)(b-1)} \geq 9/8\]
Furthermore, $h(j)$ is a monotonically increasing function, therefore $h(j) \geq 1$ for all $j \geq 1$.

Therefore:
\[f(i) \cdot \frac{1}{4} \cdot \frac{(jb+2i+2)(jb+2i+1)}{(jb+i+1)(i+1)} \geq f(i).\]
\end{proof}

Using Claim~\ref{clm:supp-helper-claim-monotonic} we have shown that:
\[E[H_j] \leq  b \cdot b \cdot {(j+2)b - 2 \choose b-1} \left(\frac{1}{2}\right)^{(j+2)b - 1}.\]

We first show that $E[H_1] \leq 1.4$ for all $b$.
\begin{lemma}
$E[H_1] \leq 1.4$ for all $b \in \mathds{N}$.
\label{clm:holes-math-help-1}
\end{lemma}
\begin{proof}
Denote by $f(b)$ the value of $E[H_1]$ for parameter $b \in \mathds{N}$. We first show that $\frac{f(b+1)}{f(b)} < 1$ for all $b\geq 12$.

\begin{align*}
\frac{f(b+1)}{f(b)} &= \frac{(b+1)^2 \cdot {3b+1 \choose b} \cdot 0.5^{3b+2}}{b^2 \cdot {3b-2 \choose b-1} \cdot 0.5^{3b-1}} \\
&=\left(\frac{b+1}{b}\right)^2 \cdot \frac{3b+1}{2b+1} \cdot \frac{3b}{2b} \cdot \frac{3b-1}{b} \cdot 0.5^3
\end{align*}
Note that $\frac{3b+1}{2b+1}$ is monotonically increasing for $b \geq 1$, and is bounded by $3/2$. Furthermore, $\frac{3b-1}{b}$ is also monotonically increasing for $b \geq 1$, and is bounded by $3$. Therefore,
\begin{align*}
\left(\frac{b+1}{b}\right)^2 \cdot \frac{3b+1}{2b+1} \cdot \frac{3b}{2b} \cdot \frac{3b-1}{b} \cdot 0.5^3 \\
&\leq \left(\frac{b+1}{b}\right)^2 \cdot \frac{6.75}{8}
\end{align*}
Finally, note that $\left(\frac{b+1}{b}\right)^2$ is a monotonically decreasing series for $b \geq 1$. For $b=12$:
\[\left(\frac{12+1}{12}\right)^2 \cdot \frac{6.75}{8} < 1.\]

Therefore, $f(b+1) \leq f(b) < 1$ for all $b \geq 12$. Lastly:
\[\max_{1\leq b \leq 12}\{f(b)\} = f(9)=1.305 < 1.4.\]

Therefore $E[H_1] < 1.4$, as required.
\end{proof}

Next, we show that $E[H_{j+1}] \leq 0.5 \cdot E[H_j]$.
\begin{lemma}
$E[H_{j+1}] \leq 0.5 \cdot E[H_j]$ for all $b \in \mathds{N}$ and $j \geq 1$.
\label{clm:holes-math-help-2}
\end{lemma}
\begin{proof}
\begin{align*}
E[H_{j+1}] &= b^2 \cdot {(j+3)b - 2 \choose b-1} \left(\frac{1}{2}\right)^{(j+3)b - 1} \\
&= \frac{1}{2} b^2 \cdot \left(\frac{1}{2}\right)^{(j+2)b - 1} \cdot \left(\frac{1}{2}\right)^{b - 1} \cdot {(j+3)b - 2 \choose b-1}
\end{align*}
We next show that:
\[\left(\frac{1}{2}\right)^{b - 1} \cdot {(j+3)b - 2 \choose b-1} \leq {(j+2)b - 2 \choose b-1},\]
which completes the proof.
\begin{align*}
&\left(\frac{1}{2}\right)^{b - 1} \cdot {(j+3)b - 2 \choose b-1} \\
&=\left(\frac{1}{2}\right)^{b - 1} \cdot \frac{((j+3)b - 2)!}{((j+2)b - 1)!\cdot(b-1)!} \\
&= \left(\frac{1}{2}\right)^{b - 1} \cdot \frac{((j+2)b - 2+b)!}{((j+1)b - 1+b)!\cdot(b-1)!} \\
&= {(j+2)b - 2 \choose b-1} \cdot \left(\frac{1}{2}\right)^{b - 1} \cdot \prod_{k=1}^b \frac{(j+2)b-2+k}{(j+1)b-1+k} \\
&= {(j+2)b - 2 \choose b-1} \cdot \prod_{k=1}^b \frac{(j+2)b-2+k}{(j+2)b-2+k+jb+k} \\ 
&\leq {(j+2)b - 2 \choose b-1} \cdot \prod_{k=1}^b 1 = {(j+2)b - 2 \choose b-1}
\end{align*}
\end{proof}






\end{document}